%% file: pricing.tex
\documentclass[11pt]{article}

\usepackage{fullpage}

\usepackage{amsfonts}
\usepackage{amssymb}
\usepackage{amstext}
\usepackage{amsmath}
\usepackage{mathtools}

\usepackage{tikz}
\usetikzlibrary{calc}
\usetikzlibrary{shapes,arrows,positioning,shadows,snakes}

\usepackage{amsthm} 

\usepackage{nameref}
\usepackage[linktocpage=true,pagebackref=true]{hyperref}
\usepackage{cleveref}

\usepackage{thmtools,thm-restate} 

\usepackage{graphicx}
\usepackage{graphics}
\usepackage{colordvi}
\usepackage{xspace}
\usepackage{algorithm}
\usepackage{algorithmicx}
\usepackage{algpseudocode}
\usepackage{url}
\usepackage{enumitem}

        {\hspace*{\fill}$\Box$\par\vspace{4mm}}

\makeatletter
\def\thmt@refnamewithcomma #1#2#3,#4,#5\@nil{%
  \@xa\def\csname\thmt@envname #1utorefname\endcsname{#3}%
  \ifcsname #2refname\endcsname
    \csname #2refname\expandafter\endcsname\expandafter{\thmt@envname}{#3}{#4}%
  \fi
}
\makeatother

\declaretheorem[numberwithin=section]{theorem}
\declaretheorem[numberlike=theorem]{lemma}
\declaretheorem[numberlike=theorem]{corollary}
\declaretheorem[numberlike=theorem]{proposition}

\declaretheorem[numberlike=theorem,refname={assumption,assumptions},Refname={Assumption,Assumptions}]{assumption}
\declaretheorem[numberlike=theorem,refname={Claim, Claims},Refname={Claim, Claims}]{claim}
\declaretheorem[numberlike=theorem,refname={hypothesis,hypothesis},Refname={Hypothesis,Hypothesis}]{hypothesis}
\declaretheorem[numberlike=theorem]{definition}

\newcommand{\BBIS}{{\sf BBIS}\xspace}

\newcommand{\sat}{{\sf SAT}\xspace}

\newcommand{\yi}{{\sc Yes-Instance}\xspace}
\renewcommand{\ni}{{\sc No-Instance}\xspace}

\newcommand{\ceil}[1]{\ensuremath{\left\lceil#1\right\rceil}}

\newcommand{\prof}{\ensuremath{\operatorname{profit}}}
\newcommand{\pay}{\ensuremath{\operatorname{pay}}}

\renewcommand{\P}{\mbox{\sf P}}
\newcommand{\NP}{\mbox{\sf NP}}

\newcommand{\ZPP}{\mbox{\sf ZPP}}
\newcommand{\polylog}[1]{\mathrm{polylog}(#1)}
\newcommand{\DTIME}{\mbox{\sf DTIME}}

\newcommand{\ZPTIME}{\mbox{\sf ZPTIME}}

\newcommand{\opt}{\mbox{\sf OPT}}

\newcommand{\set}[1]{\left\{ #1 \right\}}

\newcommand{\iset}{{\mathcal{I}}}
\newcommand{\pset}{{\mathcal{P}}}

\newcommand{\aset}{{\mathcal{A}}}
\newcommand{\cset}{{\mathcal{C}}}

\newcommand{\mset}{{\mathcal M}}

\newcommand{\be}{\begin{enumerate}}
\newcommand{\ee}{\end{enumerate}}
\newcommand{\bd}{\begin{description}}
\newcommand{\ed}{\end{description}}
\newcommand{\bi}{\begin{itemize}}
\newcommand{\ei}{\end{itemize}}
\newcommand{\bip}{\mbox{B}\xspace}
\newcommand{\bipm}{\ensuremath{\bip}}
\newcommand{\bipp}{\ensuremath{\bip}}

\renewcommand{\phi}{\varphi}

\newcommand{\poly}{\operatorname{poly}}

\newcommand{\reals}{{\mathbb R}}

\newcommand{\R}{\ensuremath{\mathbb R}}

\newcommand{\induce}[1]{\ensuremath{{\sf im}(#1)}\xspace}
\newcommand{\sinduce}[1]{\ensuremath{{\sf sim}(#1)}\xspace}
\newcommand{\sinducesigma}[2]{\ensuremath{{\sf sim}_{#1}(#2)}\xspace}

\newcommand{\SMP}{{\sf SMP}\xspace}
\newcommand{\UDP}{{\sf UDP}\xspace}
\newcommand{\semiind}{{\sf Semi-Induced Matching}\xspace}

\newcommand{\p}{{\bf p}}

\ifdefined\ShowComment

\def\danupon#1{\marginpar{$\leftarrow$\fbox{D}}\footnote{$\Rightarrow$~{\sf #1 --Danupon}}}
\def\danuponb#1{{\bf Danupon:} #1}
\def\parinya#1{\marginpar{$\leftarrow$\fbox{P}}\footnote{$\Rightarrow$~{\sf #1 --Parinya}}}
\def\bundit#1{\marginpar{$\leftarrow$\fbox{B}}\footnote{$\Rightarrow$~{\sf #1 --Bundit}}}

\else

\def\danupon#1{}
\def\danuponb#1{}
\def\parinya#1{}
\def\bundit#1{}

\fi

\title{Independent Set, Induced Matching, and Pricing: Connections and Tight (Subexponential Time) Approximation Hardnesses}

\date{}

\author{ 
	      Parinya Chalermsook\thanks{Max-Planck-Institut f\"ur Informatik, Germany. 
          Work partially done while at IDSIA, Lugano, Switzerland.
          Supported by the Swiss National Science Foundation project 200020\_144491/1}
 \and
        Bundit Laekhanukit\thanks{McGill University, Canada. 
          Supported by the Natural Sciences and
          Engineering Research Council of Canada (NSERC) grant
          no.~429598 and by Dr\&Mrs M.Leong fellowship.} 
 \and
        Danupon Nanongkai\thanks{Nanyang Technological University (NTU), Singapore. Supported in part by the following research grants: Nanyang Technological University grant M58110000, Singapore Ministry of Education (MOE) Academic Research Fund (AcRF) Tier 2 grant MOE2010-T2-2-082, and Singapore MOE  AcRF Tier 1 grant MOE2012-T1-001-094.}
}

\begin{document}


\begin{titlepage}
\maketitle
\pagenumbering{roman}
\vspace{-.5cm}
\input{abstract}

\newpage
\setcounter{tocdepth}{2}
\tableofcontents
\newpage
\listoftheorems
\end{titlepage}

\newpage
\pagenumbering{arabic}

\input{intro}

\input{overview}

\input{prelim}

\input{pcp}

\input{hardness-matching}

\input{hardness-pricing}

\input{algo}

\section{Open Problems}
There are many problems left open. The most fundamental one in algorithmic pricing (from the perspective of approximation algorithms community) is perhaps the {\em graph pricing problem} which is the $k$-hypergraph pricing problem where $k=2$. 
Currently, only a simple $4$-approximation algorithm and a hardness of $2-\epsilon$ assuming the Unique Game Conjecture, are known~\cite{KhandekarKMS09}. 
It is interesting to see if the techniques in this paper can be extended to an improved hardness (which will likely to require even tighter connections). 
Other interesting problems that seem to be unachievable using the current techniques are the {\em Stackelberg network pricing}, {\em Stackelberg spanning tree pricing}, and {\em tollbooth pricing} problems. 

Another interesting question is whether our techniques can be used to make a progress in the parameterized complexity domain. In particular, it was conjectured in~\cite{Marx13,ChitnisTK13} that 
the maximum independent set problem parameterized by the size of the solution does not admit a fixed parameter tractable approximation ratio $\rho$ for any function $\rho$.  
It might be also interesting to improve our $2^{n^{1-\epsilon}/ r^{1+\epsilon}}$ time lower bound for $r$-approximating the maximum independent set and induced matching problems, e.g., to $2^{(n\cdot \polylog r)/(r\cdot \polylog n)}$. 
Other (perhaps less important) open problems also remain: (1) Is the ETH necessary in proving the lower bound of this problem? For example, can we get a better approximation guarantee for the $k$-hypergraph pricing problem if there is a subexponential-time algorithm for solving SAT (see, e.g.,~\cite{CyganDLMNOPSW12} for similar questions in the exact algorithm domain)? (2) Is it possible to obtain an $r$-approximation algorithm in $2^{n/r}$ time for the maximum induced matching problem in general graphs? 


%

\bibliographystyle{abbrv}
\bibliography{pricing}

\newpage
\appendix
\section*{Appendix}

\input{fglss}


\end{document}

%% file: abstract.tex

\begin{abstract}
We present a series of almost settled inapproximability results for three fundamental problems. The first in our series is the {\em subexponential-time} inapproximability of the {\em maximum independent set} problem, a question studied in the area of {\em parameterized complexity}. The second is the hardness of approximating the {\em maximum induced matching} problem on bounded-degree bipartite graphs. The last in our series is the tight hardness of approximating the {\em $k$-hypergraph pricing} problem, a fundamental problem arising from the area of {\em algorithmic game theory}. In particular, assuming the Exponential Time Hypothesis, our two main results are:

\begin{itemize}
\item For any $r$ larger than some constant, any $r$-approximation algorithm for the maximum independent set problem must run in at least $2^{n^{1-\epsilon}/r^{1+\epsilon}}$ time.
This nearly matches the upper bound of $2^{n/r}$ ~\cite{CyganKPW08}. It also improves some hardness results in the domain of parameterized complexity (e.g., \cite{EscoffierKP12-FPT-approx,ChitnisTK13})\danupon{Should we mention \cite{ChitnisTK13}. Should we let them know first?}.

\item For any $k$ larger than some constant, there is no polynomial time $\min \set{k^{1-\epsilon}, n^{1/2-\epsilon}}$-approximation algorithm for the $k$-hypergraph pricing problem , where $n$ is the number of vertices in an input graph. This almost matches the upper bound of  $\min \set{O(k), \tilde O(\sqrt{n})}$  (by Balcan and Blum \cite{BB07} \danupon{in \cite{BB07} $\rightarrow$ \cite{BB07}} and an algorithm in this paper).
\end{itemize}

We note an interesting fact that, in contrast to $n^{1/2-\epsilon}$ hardness for polynomial-time algorithms\danupon{``for polynomial-time algorithms'' is added}, the $k$-hypergraph pricing problem admits $n^{\delta}$ approximation for any $\delta >0$ in quasi-polynomial time. This puts this problem\danupon{pricing problem $\rightarrow$ this problem} in a rare approximability class in which approximability thresholds can be improved significantly by allowing algorithms to run in quasi-polynomial time.

The proofs of our hardness results rely on unexpectedly\danupon{several unexpected $\rightarrow$ remove several} tight connections between the three problems. First, we establish a connection between the first and second problems by proving a new graph-theoretic property related to an {\em induced matching number of dispersers}\danupon{made ``disperser'' italic}. Then, we show that the $n^{1/2-\epsilon}$ hardness of the last problem follows from nearly tight {\em subexponential time} inapproximability of the first problem, illustrating a rare application of the second type of inapproximability result to the first one.  Finally, to prove the subexponential-time inapproximability of the first problem, we construct a new PCP with several properties; it is sparse and has nearly-linear size, large degree, and small free-bit complexity. Our PCP requires no ground-breaking ideas but rather a very careful assembly of the existing ingredients in the PCP literature.  
\end{abstract}

%% file: intro.tex
\section{Introduction}
\label{sec:intro}

\danupon{To discuss: Where should we define the semi-induced matching problem?}

This paper presents results of two kinds, lying in the intersections between approximation algorithms and other subareas of theoretical computer science. 
The first kind of our results is a tight hardness of approximating the $k$-hypergraph pricing problem in polynomial time. 
This problem arose from the area of algorithmic game theory, and its several variants have recently received attentions from many researchers (see, e.g., \cite{Rusmevichientong2003,RusmevichientongRG06,GuruswamiHKKKM05,BriestK06,BB07,Briest08,PopatW12,ChalermsookCKK12})\danupon{Should we make this sentence a bit less agressive?}.  
It has, however, resisted previous attempts to improve approximation ratio given by simple algorithms\danupon{algorithm (singular)? ``the'' approximation ratio? Maybe we should be more explicit on what algorithms we're talking about. I guess they're from \cite{BriestK06,BB07}, right?}. 
Indeed, no sophisticated algorithmic techniques have been useful in attacking the problem in its general form. 
The original motivation of this paper is to show that those simple algorithms are, in fact, the best one can do under a reasonable complexity theoretic assumption. 
In showing this, we devise a new reduction from another problem studied in discrete mathematics and networking called the {\em maximum bipartite induced matching} problem. 
Our reduction, unfortunately, blows up the instance size {\em exponentially}, and apparently this blowup is unavoidable (this claim will be discussed precisely later). 
Due to the exponential blowup of our reduction, showing a {\em tight} polynomial-time hardness of approximating the maximum bipartite induced matching problem is not enough for settling the complexity of the pricing problem. 
What we need is, roughly speaking, the hardness of approximation result that is tight even for subexponential time approximation algorithms, i.e., proving the lower bound on the approximation ratio that any subexponential time algorithms can achieve.

This motivates us to prove the second type of results: {\em hardness of subexponential-time approximation}. 
The subject of subexponential time approximation and the closely related subject of {\em fixed-parameter tractable (FPT) approximation} have been recently studied in the area of parameterized complexity (e.g., \cite{CyganKPW08,FellowGMS12,EscoffierKP12-FPT-approx,ChitnisTK13}). 
Our main result of this type is a {\em sharp trade-off} between the running time and approximation ratio for the bipartite induced matching problem, and since our proof crucially relies on the hardness construction for the maximum independent set problem\danupon{Should we use abbreviation?}, we obtain a sharp trade-off for approximating the maximum independent set problem as a by-product. 
The maximum independent set problem is among fundamental problems studied in both approximation algorithms and FPT literature (since it is W[1]-hard), and it is of interest to figure out its subexponential-time approximability. 
Our trade-off result immediately answers this question, improves previous results in \cite{EscoffierKP12-FPT-approx,ChitnisTK13} and nearly matches the upper bound in \cite{CyganKPW08}.



The main contributions of this paper are the {\em nearly tight connections} among the aforementioned problems (they are tight in the sense that any further improvements would immediately refute the Exponential Time Hypothesis (ETH)), which essentially imply the nearly tight (subexponential time) hardness of approximation for all of them.
Interestingly, our results also illustrate a rare application of the subexponential-time inapproximability to the inapproximability of polynomial-time algorithms. 
The key ideas of our hardness proofs are simple and algorithmic even though it requires a non-trivial amount of work to actually implement them\danupon{Should we change this sentence to make it less aggressive?}.
%
%

Finally, we found a rather bizarre phenomenon of the $k$-hypergraph pricing problem (when $k$ is large) in the quasi-polynomial time regime.  
While both induced matching, independent set and many other natural combinatorial optimization problems\danupon{Is it correct to use ``both''?} do not admit much better approximation ratios in quasi-polynomial time (e.g., $n^{1-\epsilon}$ hardness of approximating the independent set and bipartite induced matching problem still hold against quasi-polynomial time algorithms), the story is completely different for the pricing problem: 
That is,\danupon{I don't think we need ``that is'' since the sentence is after ``:''} the pricing problem admits $n^{\delta}$ approximation in quasi-polynomial time for any $\delta >0$, even though it is $n^{1/2-\epsilon}$ hard against polynomial-time approximation algorithms.  
This contrast puts the pricing problem in a rare approximability class in which polynomial time and quasi-polynomial time algorithms' performances are significantly different.\danupon{I removed the next paragraph. See the commented paragraph.}\danupon{For later: Possible new organization: Create three subsections with names like ``Problems \& Results 1: $k$-hypergraph pricing"}


\subsection{Problems}

\paragraph{$k$-Hypergraph Pricing} In the {\em unlimited supply $k$-hypergraph vertex pricing} problem \cite{BB07,BriestK06}, we are given a weighted $n$-vertex $m$-edge $k$-hypergraph (each hyperedge contains at most $k$ vertices) modeling the situation where consumers (represented by hyperedges) with budgets (represented by weights of hyperedges) have their eyes on at most $k$ products (represented by vertices). The goal is to find a price assignment that maximizes the revenue. In particular, there are two variants of this problem with different consumers' buying rules. 
In the {\em unit-demand pricing problem} (\UDP), we assume that each consumer (represented by a hyperedge $e$) will buy the cheapest vertex of her interest if she can afford it. In particular, for a given hypergraph $H$ with edge weight $w:E(H)\rightarrow \reals_{\geq 0}$ (where $\mathbb{R}_{\geq 0}$ is the set of non-negative reals), our goal is to find a price function $p: V(H) \rightarrow \mathbb{R}_{\geq 0}$  to maximize 
\[
\prof_{H, w}(p) = \sum_{e\in E(H)} \pay_e(p) ~~~\mbox{where } \pay_e(p)=
\begin{cases}
\min_{v\in e} p(v) & \mbox{if } \min_{v\in e} p(v)\leq w(e),\\
0 & \mbox{otherwise.} 
\end{cases}
\]
The other variation is the {\em single-minded pricing problem} (\SMP), where we assume that each consumer will buy {\em all} vertices if she can afford to; otherwise, she will buy nothing. Thus, the goal is to maximize
\[
\prof_{H, w}(p) = \sum_{e\in E(H)} \pay_e(p) ~~~\mbox{where } \pay_e(p)=
\begin{cases}
\sum_{v\in e} p(v) & \mbox{if } \sum_{v\in e} p(v)\leq w(e),\\
0 & \mbox{otherwise.} 
\end{cases}
\]
The pricing problem naturally arose in the area of algorithmic game theory 
and has important connections to algorithmic mechanism design (e.g., \cite{BalcanBHM05,ChawlaHK07}). 
%
Its general version (where $k$ could be anything) was introduced by Rusmevichientong et al. \cite{Rusmevichientong2003,RusmevichientongRG06}, and the $k$-hypergraph version (where $k$ is thought of as some constant) was first considered by Balcan and Blum \cite{BB07}. (The special case of $k=2$ has also received a lot of attention \cite{BB07,ChalermsookKLN12,KhandekarKMS09,PopatW12}.)
%
%
%
There will be two parameters of interest to us, i.e., $n$ and $k$. 
Its current approximation upper bound is $O(k)$ \cite{BriestK06,BB07} while its lower bound is $\Omega(\min(k^{1/2-\epsilon}, n^\epsilon))$ \cite{Briest08,ChalermsookCKK12,ChalermsookLN-SODA13}.
%
%
%
%
\danupon{I commented out the last few sentences of this paragraph since it repeats what we have in the result section.}

\paragraph{Bipartite Induced Matching} 
Informally, an induced matching of an undirected unweighted graph $G$ is a matching $M$ of G such that no two edges in $M$ are joined by an edge in $G$.
%
To be precise, let $G=(V, E)$ be any undirected unweighted graph. 
An {\em induced matching} of $G$ is the set of edges $\mset \subseteq E(G)$ such that $\mset$ is a matching and 
for any distinct edges $uu', vv' \in \mset$, $G$ has none of the edges in $\set{uv, uv', u'v, u'v'}$\danupon{I removed ``no two edges in $\mset$ are joined by an edge in $G$, i.e.,'' to avoid repetition. I also added ``distinct''}.
The {\em induced matching number} of $G$, denoted by $\induce{G}$, is the cardinality of the maximum-cardinality induced matching of $G$.
%
Our goal is to compute $\induce{G}$ of a bipartite graph $G$.
%
%


The notion of induced matching has naturally arisen in discrete mathematics and computer science. It is, for example, studied as the ``risk-free'' marriage problem in \cite{StockmeyerV82} and is a subtask of finding a {\em strong edge coloring}. This problem and its variations also have connections to various problems such as maximum feasible subsystem \cite{ElbassioniRRS09,ChalermsookLN-SODA13},  maximum expanding sequence \cite{BriestK11}, storylines extraction \cite{KumarMS04} and network scheduling, gathering and testing (e.g., \cite{EvenGMT84,StockmeyerV82,JooSSM10,Milosavljevic11,BonifaciKMS11}).
%
In general graphs, the problem was shown to be \NP-complete in \cite{StockmeyerV82,Cameron89} and was later shown in \cite{ChlebikC06} 
to be hard to approximate to within a factor of $n^{1-\epsilon}$ and $d^{1-\epsilon}$ unless $\P = \NP$, where $n$ is the number of vertices and $d$ is the maximum degree of a graph.
In bipartite graphs, assuming $\P \neq \NP$, the maximum induced matching problem was shown to be $n^{1/3-\epsilon}$-hard to approximate in \cite{ElbassioniRRS09}. Recently, we \cite{ChalermsookLN-SODA13} showed its tight hardness of $n^{1-\epsilon}$ (assuming $\P \neq \NP$) and a hardness of $d^{1/2-\epsilon}$ on $d$-degree-bounded bipartite graphs. This hardness leads to tight hardness of several other problems.
\danupon{I shortened the rest of this paragraph (commented) since its right place is the result section}
In this paper, we improve the previous hardness to a tight $d^{1-\epsilon}$ hardness, as well as extending it to a tight approximability/running time trade-off for subexponential time algorithms. 


\paragraph{Independent Set} 


Given a graph $G=(V,E)$, a set of vertices $S\subseteq V$ is {\em independent} (or {\em stable}) in $G$ if and only if $G$ has no edge joining any pair of vertices $x,y\in S$. 
In the maximum independent set problem, we are given an undirected graph $G=(V,E)$, and the goal is to find an independent set $S$ of $G$ with maximum size.
%
Hardness results for the maximum independent set problem heavily rely on developments in the PCP literature. 
The connection between the maximum independent set problem and the {\em probabilistic checkable proof system} (PCP) was first discovered by Feige~et~al.~\cite{FGLSS96} who showed that the maximum independent set problem is hard to approximate to within a factor of $2^{\log^{1-\epsilon}n}$, for any $\epsilon>0$, unless  $\NP\subseteq\DTIME(n^{\polylog{n}})$.
The inapproximability result has been improved by Arora and Safra~\cite{AroraS98} and Arora~et~al.~\cite{AroraLMSS98}, leading to a polynomial hardness of the problem~\cite{AroraLMSS98}. 
Later, Bellare and Sudan~\cite{BellareS94} introduced the notion of the {\em free-bit complexity} of a PCP and showed that, given a PCP with logarithmic randomness and free-bit complexity $f$, the maximum independent set problem is hard to approximate to within a factor of $n^{1/(1+f)-\epsilon}$, for all $\epsilon>0$, unless $\NP=\ZPP$. 
%
There, Bellare and Sudan~\cite{BellareS94} constructed a PCP with free-bit complexity $f=3+\delta$, for all $\delta>0$, thus proving the hardness of $n^{1/4-\epsilon}$ for the maximum independent set problem.
The result has been subsequently improved by Bellare~et~al. in \cite{BellareGS98} who gave a construction of a PCP with free-bit complexity $f=2+\delta$.
Finally, H{\aa}stad~\cite{Hastad96} constructed a PCP with arbitrary small free-bit complexity $f>0$, thus showing 
the tight hardness (up to the lower order term) of $n^{1-\epsilon}$, for all $\epsilon>0$, for the maximum independent set problem.
A PCP with optimal free-bit complexity was first constructed by Samorodnitsky and Trevisan~\cite{ST00}.
The PCP of Samorodnitsky and Trevisan in \cite{ST00} has imperfect completeness, and this has been improved by H{\aa}stad and Khot~\cite{HastadK05} to a PCP that has both perfect completeness and optimal free-bit complexity.
Recently, Moshkovitz and Raz~\cite{MR10} gave a construction of a projective 2-query PCP with nearly-linear size, which can be combined with the result of Samorodnitsky and Trevisan~\cite{ST00} to obtain a PCP with nearly-linear size and optimal free-bit complexity.
The soundness of a PCP with optimal free-bit complexity was improved in a very recent  breakthrough result of Chan~\cite{Chan12}.\bundit{{\bf Possibly Footnote:}
In fact, the nearly-linear size PCP of Moshkovitz and Raz~\cite{MR10} can also be combined with the PCP constructed by Chan, which thus gives a PCP with nearly-linear size, optimal free-bit complexity and strong soundness. But, we want to attribute the result to Samorodnitsky and Trevisan since they are the first who obtained such a result.
}
\bundit{I added a reference to Zuckerman's result}
The complexity assumption of early tight hardness results of the maximum independent set problem (e.g., \cite{Hastad96}) is $\NP\neq \ZPP$ because of a random process in the PCP constructions. This process was derandomized in \cite{Zuckerman07} by Zuckerman, thus proving tight hardness under the $\P\neq\NP$ assumption.

\subsection{Our Results}

\begin{table} 
\centering
{\footnotesize
\begin{tabular}{l|l|c|c}
\hline
{\bf Problem} &  & {\bf Upper}& {\bf Lower} \\
\hline
$k$-hypergraph pricing & Previous & $O(k)$ \cite{BB07,BriestK06} & $\Omega(k^{1/2-\epsilon})$ and $\Omega(n^{\delta})$ \cite{Briest08,ChalermsookCKK12,ChalermsookLN-SODA13}\\
(polynomial time) & This paper & $O(\min(k, (n\log n)^{1/2}))$ & $\Omega(\min(k^{1-\epsilon}, n^{1/2-\epsilon}))$\\
\hline
$k$-hypergraph pricing & Previous & - & -\\
(quasi-polynomial time) & This paper &  $n^{\delta}$-approx. algo in & $n^{\delta}$-approx. algo requires \\ && $O(2^{(\log{m})^{\frac{1-\delta}{\delta}}\log\log{m}}\poly(n,m))$-time & $\Omega(2^{(\log m)^{\frac{1-\delta-\epsilon}{\delta}}})$  time \\
\hline
Independent set  & Previous & $O(2^{n/r}\poly(n))$ time \cite{CyganKPW08}& $\Omega(2^{n^\delta/r})$ time \cite{ChitnisTK13}\\
(subexpo.-time $r$-approx. algo) & This paper & - & $\Omega(2^{n^{1-\epsilon}/r^{1+\epsilon}})$ time\\
\hline
Induced matching on  & Previous & $O(d)$ (trivial) & $\Omega(d^{1/2-\epsilon})$ \cite{ChalermsookLN-SODA13}\\
$d$-deg.-bounded bip. graphs & This paper & - & $\Omega(d^{1-\epsilon})$ \\
(polynomial time) & & &\\
\hline
Induced matching on bip. graphs & Previous & - & -\\
(subexpo.-time $r$-approx. algo) & This paper & $O(2^{n/r}\poly(n))$ time & $\Omega(2^{n^{1-\epsilon}/r^{1+\epsilon}})$ time \\
\hline
\end{tabular}
}
\caption{Summary of results.} 
\label{table:results}
\end{table}

We present several tight hardness results both for polynomial and subexponential-time algorithms, as summarized in \Cref{table:results}. Most our results rely on a plausible complexity theoretic assumption stronger than $\P \neq \NP$, namely, {\em Exponential Time Hypothesis} (ETH), which, roughly speaking, states that SAT cannot be decided by any subexponential time algorithm (see \Cref{sec:prelim} for detail).

Our first result, which is our original motivation, is the tight hardness of approximating the $k$-hypergraph pricing problems in polynomial time. These problems (both \UDP and \SMP) are known to be $O(k)$-approximable \cite{BB07,BriestK06} and the hardness of $\Omega(k^{1/2-\epsilon})$\danupon{I changed this from $\Omega(k^\delta)$.} and $\Omega(n^\delta)$, for some constant $\delta>0$, are known based on the assumption about hardness of refuting a random 3SAT formula~\cite{Briest08}. 
%
%
%
\danupon{I changed the next two sentence a bit to incorporate the commented part from previous subsection.}
A series of recent results leads to a disagreement on the right approximability thresholds of the problem. On one hand, the current best approximation algorithm is so simple that one is tempted to believe that a more sophisticated idea would immediately give an improvement on the approximation ratio. On the other hand, no algorithmic approach could go beyond the barrier of $O(k)$ so far, thus leading to a belief that $\Omega(k^{1-\epsilon})$ and $\Omega(n^{1-\epsilon})$ hardness should hold.
%
%
In this paper, we settle the approximability status of this problem. 
Somewhat surprisingly, the right hardness threshold of this problem turns out to lie somewhere between the two previously believed numbers: the believed hardness of $\Omega(k^{1-\epsilon})$ was correct but {\em only for} $k=O(n^{1/2})$.
%
\begin{theorem}[Polynomial-time hardness of $k$-hypergraph pricing; proof in \Cref{sec:hardness_pricing}]\label{thm:hardness pricing}
The $k$-hypergraph pricing problems (both \UDP and \SMP) are $\Omega(\min (k^{1-\epsilon}, n^{1/2-\epsilon}))$ hard to approximate in polynomial time unless the ETH is false\footnote{The $k^{1-\epsilon}$ hardness only requires $\NP \neq \ZPP$ when $k$ is constant.}.  
Moreover, they are $\tilde O(\min(k, (n\log n)^{1/2}))$-approximable in polynomial time. 
\end{theorem} 

The main ingredient in proving \Cref{thm:hardness pricing} is proving tight hardness thresholds of {\em subexponential-time} algorithms for the maximum independent set and the maximum induced matching problems in $d$-degree-bounded bipartite graphs\footnote{We note that, indeed, the connection to the pricing problem is via a closely related problem, called the {\em semi-induced matching} problem, whose hardness follows from the same construction as that of the maximum induced matching problem; see \Cref{sec:prelim}.}\danupon{This footnote was originally a sentence in the previous section.}. Besides playing a crucial role in proving \Cref{thm:hardness pricing}, these results are also of an independent interest. 
Our first subexponential-time hardness result is for the maximum independent set problem. 
While the polynomial-time hardness of this problem has been almost settled, the question whether one can do better in subexponential time has only been recently raised in the parameterized complexity community. Cygan et al.~\cite{CyganKPW08} and Bourgeois et al.~\cite{BourgeoisEP11-IndependentSet} independently observed that better approximation ratios could be achieved if subexponential running time is allowed\danupon{Are these the correct citations? If so, we should change the abstract accordingly.}. 
In particular, they showed that an $r$ approximation factor can be obtained in  $O(2^{n/r}\poly(n))$ time. Recently, Chitnis et al. \cite{ChitnisTK13} showed that, assuming the ETH, an $r$-approximation algorithm requires $\Omega(2^{n^\delta/r})$ time, for some constant $\delta>0$\footnote{We note that their real statement is that for any constant $\epsilon>0$, there is a constant $F > 0$ depending on $\epsilon$ such that CLIQUE (or equivalently the independent set problem) does not have an $n^{\epsilon}$-approximation in  $O(2^{\opt^{F/\epsilon}}\cdot \poly(n))$ time. Their result can be translated into the result we state here.}\danupon{Is the footnote correct?}. Our hardness of the maximum independent set problem improves upon the lower bound of Chitnis et al. and essentially matches the upper bounds of Cygan et al. and Bourgeois et al.
\begin{theorem}[Subexponential-time hardness of independent set; proof in \Cref{sec: hardness of induced matching}]\label{thm:hardness independent set}
Any $r$ approximation algorithm for the maximum independent set problem must run in time 
$2^{n^{1-\epsilon}/r^{1+\epsilon}}$ unless the ETH is false.  
\end{theorem} 
An important immediate step in using \Cref{thm:hardness independent set} to prove \Cref{thm:hardness pricing} is proving the subexponential-time hardness of the induced matching problem on $d$-degree-bounded bipartite graphs. The polynomial-time hardness of $n^{1-\epsilon}$ for this problem has only been resolved recently by the authors of this paper in \cite{ChalermsookLN-SODA13}, where we also showed a hardness of $d^{1/2-\epsilon}$ when the input graph is bipartite of degree at most $d$. In this paper, we improve this bound to $d^{1-\epsilon}$ and extend the validity scope of the results to subexponential-time algorithms. 
The latter result is crucial for proving \Cref{thm:hardness pricing}.\danupon{So, the problem on $d$-regular graph is still open, right?}

\begin{theorem}[Hardness of induced matching on $d$-degree-bounded bipartite graphs; proof in \Cref{sec: hardness of induced matching}]\label{thm:hardness induced matching}
Let $\epsilon >0$ be any constant. 
For any $d\geq c$, for some constant $c$ (depending on $\epsilon$), there is no $d^{1-\epsilon}$ approximation algorithm for the maximum induced matching problem in $d$-degree-bounded bipartite graphs unless $\NP = \ZPP$. Moreover, any $r$ approximation algorithm for the maximum induced matching problem in bipartite graphs must run in time $2^{n^{1-\epsilon}/r^{1+\epsilon}}$ unless the ETH is false.\danupon{on graph or in graphs?}  
\end{theorem}  

%
\danupon{Is Chapter 4 in \url{http://www.dur.ac.uk/konrad.dabrowski/dabrowskithesis.pdf} useful to us?}

Finally, we note an interesting fact that, while the polynomial-time hardness of the $k$-hypergraph pricing problem follows from the hardness of the independent set and the bipartite induced matching problems, its subexponential time approximability is quite different from those of the other two problems. 
In particular, if we want to get an approximation ratio of $n^{\epsilon}$ for some constant $\epsilon>0$. \Cref{thm:hardness independent set,thm:hardness induced matching} imply that we still need subexponential time to achieve such an approximation ratio. In contrast, we show that, for the case of the $k$-hypergraph pricing problem, such an approximation ratio can be achieved in quasi-polynomial time. 

\begin{theorem}[Quasi-polynomial time $n^{\delta}$-approximation scheme; proof in \Cref{sec:algo}]
For the $k$-hypergraph pricing problem and any constant $\delta>0$, there is an algorithm that gives an approximation ratio of $O(n^{\delta})$ and runs in time $O(2^{(\log{m})^{\frac{1-\delta}{\delta}}\log\log{m}}\poly(n,m))$.
\end{theorem}

We also prove (in \Cref{sec:hardness_pricing}) that the above upper bound is tight. 

\subsection{Techniques}\label{sec:techniques}

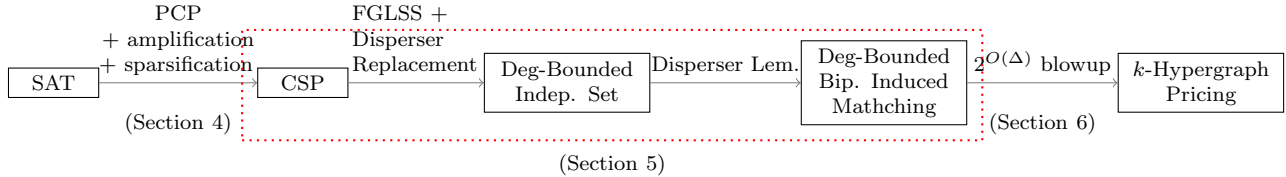
\begin{figure*}
\begin{center}
\begin{tikzpicture}[auto, node distance=0.1cm]
    \scriptsize
    \tikzstyle{mybox}=[rectangle, draw=black, text centered, text=black, text width=2cm]
%
    \node (sat) [rectangle, draw=black, text centered, text=black, text width=1cm] {SAT};
    \node (csp) [rectangle, draw=black, text centered, text=black, text width=1cm, right=2.1 of sat] {CSP};
    \node (indep) [mybox, right=1.8 of csp] {Deg-Bounded Indep. Set};
    \node (induce) [mybox, right=2 of indep] {Deg-Bounded Bip. Induced Mathching};
    \node (pricing) [mybox, right=2 of induce] {$k$-Hypergraph Pricing};
    \draw[->, thin, gray]
       (sat) to 
       node (PCP-text) [black, text width=2.1cm]{\scriptsize \centering PCP \\ + amplification\\ + sparsification}
       node[black, below = 0.3 of PCP-text]
         {(Section~\ref{sec: PCP})}
       (csp);
    \draw[->, thin, gray]
       (csp) to 
       node (FGLSS-text) [black,text width=1.7cm] {\scriptsize \centering FGLSS + Disperser Replacement}
       (indep);
    \draw[->, thin, gray]
       (indep) to
       node (Disperser-text) [black] {\scriptsize Disperser Lem.} 
       (induce);
    \draw[->, thin, gray]
       (induce) to
       node (pricing-reduction-text) [black] {\scriptsize $2^{O(\Delta)}$ blowup}
       node[black, below = 0.3 of pricing-reduction-text]
              {(Section~\ref{sec:hardness})}  
       (pricing);
    \draw [red,thick,dotted]
        ($(csp.north west)+(-0.2,0.5)$)
        rectangle ($(induce.south east)+(0.2,-0.2)$)
        node[black, below = 0.8, midway]
          {(Section~\ref{sec: hardness of induced matching})};
\end{tikzpicture}
\end{center}
\caption{Hardness proof outline (more detail in \Cref{sec:overview}).}\label{fig:summary}
\end{figure*}

\danupon{The biggest problem of this section that we should discuss is that it's not consistent with the next section.}

The key ingredients in our proofs are {\em tight} connections between 3SAT, independent set, induced matching, and pricing problems. Our reductions are fairly tight in the sense that their improvement would violate the ETH. They are outlined in \Cref{fig:summary} (see \Cref{sec:overview} for a more comprehensive overview). 
Among several techniques we have to use, the most important new idea is a {\em simple new property of dispersers}. We show that an operation called {\em bipartite double cover} will convert a disperser $G$ to another bipartite graph $H$ whose maximum induced matching size is essentially the same as the size of maximum independent set of $G$. This property crucially relies on the fact that $G$ is a disperser.
We note that the disperser that we study here is not new -- it has been around for at least two decades (e.g., \cite{CohenW89,Sipser88}) -- but the property that we prove is entirely new. We provide the proof idea of this new property in \Cref{sec:disperser overview}. 

Our hardness proof of the maximum induced matching problem in $d$-degree-bounded graphs is inspired by the (implicit) reduction of Briest \cite{Briest08} and the previous (explicit) reduction of ours \cite{ChalermsookLN-SODA13}. These previous reductions are not tight in two aspects: (i) they do not result in a tight $d^{1-\epsilon}$ hardness for the maximum bipartite induced matching problem (Briest's reduction in \cite{Briest08} only gives a hardness of $d^{\delta}$ for some $\delta >0$, and our previous reduction in \cite{ChalermsookLN-SODA13} only gives a hardness of $d^{1/2 -\epsilon}$), and (ii) they do not give any hardness for subexponential-time approximation algorithms. 
%
%
In this paper, we make use of additional tools to improve these two aspects to get a tight reduction. The first tool is the new property of the disperser as we have discussed. The second tool is a new PCP which results from carefully combining known techniques and ideas from the PCP literature.

Our PCP requires an intricate combination of many known properties: That is, it must be sparse, as well as, having nearly-linear size, large degree, and small free-bit complexity. 
We explain some of the required properties here.
The sparsity and the size of the PCP are required in order to boost hardness of the $k$-hypergraph pricing problem to $n^{1/2 -\epsilon}$ (without these, we would not go beyond $n^{\delta}$ for some small $\delta$). 
The large degree of the PCP is needed to ensure that our randomized construction is successful with high probability. Finally, the small free-bit complexity is needed to get the $d^{1-\epsilon}$ hardness for the maximum bipartite induced matching and the independent set problems in $d$-degree-bounded graphs; this is the same idea as those used in the literature of proving hardness of the maximum independent set problem.  

Our proof of the hardness of the $k$-hypergraph pricing problem from the hardness of the bipartite induced matching problem is inspired by the previous proofs in \cite{BriestK11,ChalermsookCKK12}. Both previous proofs require a hardness of some special cases of the bipartite induced matching problem (e.g., \cite{ChalermsookLN-SODA13} requires that the input instance is a result of a certain graph product) in order to derive the hardness of the $k$-hypergraph pricing problem. In this paper, we provide insights that lead to a reduction that simply exploits the fact that the input graph has bounded degree, showing a clearer connection between the two problems. 

\subsection{Organization}
We give an overview of our hardness proof in \Cref{sec:overview}.
The hardness results are proved in three consecutive sections (i.e., Sections~\ref{sec: PCP}, \ref{sec: hardness of induced matching} and \ref{sec:hardness}). We first present our PCP in Section~\ref{sec: PCP}. The connection between SAT and the maximum independent set and bipartite induced matching problems is presented in Section~\ref{sec: hardness of induced matching}. Finally, Section~\ref{sec:hardness} shows a connection between the bipartite induced matching and the $k$-hypergraph pricing problems.  
A new approximation algorithm is shown in Section~\ref{sec:algo}. 

%% file: overview.tex
\section{Overview of the Hardness Proof}
\label{sec:overview}

Our proofs involve many parameters. To motivate the importance of each parameter and the need of each tool, we give the top-down overview of our proofs in an informal manner. The presentation here will be imprecise, but we hope that this would help readers understanding our proofs.

\subsection{The Main Reduction: SAT $\rightarrow$ pricing} At a very high level, our proof makes the following connection between \sat and the $h$-hypergraph pricing problem. 
%
%
%
We give a reduction that transforms an $N$-variable SAT formula into a $h$-hypergraph pricing instance $(\cset, \iset)$ with the following parameters: the number of items is $|\iset| = Nh$, the number of consumers is $|\cset| = 2^{h} \poly(N)$, and the hardness gap is $h$. The following concludes the interaction between parameters. 
%
%


{\footnotesize
\begin{equation}
\underbrace{\mbox{$N$-variable \sat}}_{\mbox{no gap}} 
  \quad \implies \quad  
  \underbrace{\mbox{$h$-hypergraph pricing with $|\iset| = Nh$, $|\cset| = 2^{h} \poly N$}}_{\mbox{hardness gap $h$}}
\label{eq:sat_to_pricing}
\end{equation}
}

The running time of our reduction is small, i.e., $\poly(|\cset|)$, and thus can be ignored for now. The hardness gap of $h$ means that any algorithm that could solve the pricing problem with an approximation factor of $o(h)$ would be able to solve \sat (optimally) within the same running time. Our main task is to obtain a reduction as in \Cref{eq:sat_to_pricing}, which we explain further in \Cref{sec:sat_indep_pricing}. Below we sketch how the reduction in \Cref{eq:sat_to_pricing} gives a hardness of the pricing problem.

\paragraph{Hardness of Pricing} It follows from \Cref{eq:sat_to_pricing} that, for $h= o(N)$, any $o(h)$-approximation algorithm for the pricing problem that runs in $\poly(|\iset|, |\cset|)$-time would imply a subexponential-time ($(2^{o(N)}\poly N)$-time) algorithm that solves \sat, thus breaking the ETH. So, we obtain a hardness of $h$ for any $h = o(N)$, assuming the ETH; this hardness can be written as $\min(h, \sqrt{\iset})$ since $|\iset| = Nh$.\danupon{Should we also explain the subexponential-time approximation hardness of pricing here?}

%
%

%
%



\subsection{Intermediate Steps: SAT $\rightarrow$ independent set/induced matching $\rightarrow$ pricing} \label{sec:sat_indep_pricing}

The main reduction (\Cref{eq:sat_to_pricing}) goes through two intermediate problems: (1)~the maximum independent set problem and (2)~the maximum induced matching problem on bipartite graphs. Both problems are considered in the case where an input graph has maximum degree $\Delta$. 
There are two intermediate steps. The first step transforms any $N$-variable \sat formula into a graph $G$, instance of the maximum independent set/induced matching problem such that the number of vertices $|V(G)| = N \Delta$, and the hardness gap $\Delta$.

%

{\footnotesize
\begin{equation}
\mbox{$N$-variable \sat} 
  \quad \implies \quad
  \underbrace{\mbox{$\Delta$-degree bounded indep. set with $|V(G)| = N \Delta$}}_{\mbox{hardness gap $\Delta$}}
\label{eq:sat_to_indepset}
\end{equation}
}

The second step transforms the independent set problem to the induced matching problem. The crucial thing of our transformation is that it does not blowup any parameter. 
{\scriptsize
\begin{equation}
\underbrace{\mbox{$\Delta$-degree bounded indep. set with $|V(G)| = N \Delta$}}_{\mbox{hardness gap $\Delta$}}
  ~ \implies ~
  \underbrace{\mbox{$\Delta$-degree bounded induced matching with $|V(G)| = N \Delta$}}_{\mbox{hardness gap $\Delta$}}
\label{eq:indepset_to_induce}
\end{equation}
}

The third step transforms the instance of the maximum independent set/induced matching problem to the $\Delta$-hypergraph pricing problem with the following parameters: the number of items is $|\iset| = |V(G)|$, the number of consumers is $|\cset| = 2^{\Delta} \poly |V(G)|$, and the hardness gap is $\Delta$.
%
%
{\footnotesize
\begin{equation}
  \underbrace{\mbox{$\Delta$-degree bounded ind. matching}}_{\mbox{hardness gap $\Delta$}}\implies
\underbrace{\mbox{$\Delta$-hyp. pricing with $|\iset| = |V(G)|$, $|\cset| = 2^{\Delta} \poly |V(G)|$}}_{\mbox{hardness gap $\Delta$}}
\label{eq:indepset_to_pricing}
\end{equation}
}


%

In \Cref{sec:intro:sat_indep,sec:intro:indep_induce,sec:intro:induce_pricing}, we explain the reductions in \Cref{eq:sat_to_indepset,eq:indepset_to_induce,eq:indepset_to_pricing} in more detail. Below we sketch how the reduction in \Cref{eq:sat_to_indepset} implies sub-exponential time inapproximability results for the maximum independent set and induced matching problems.

%
\paragraph{Hardness of Max. Independent Set and Max. Induced Matching} Suppose that there is a $o(\Delta)$-approximation $2^{o(|V(G)|/\Delta)}$-time algorithm.  This algorithm will solve \sat since the reduction in \Cref{eq:sat_to_indepset} gives a hardness gap of $\Delta$. Moreover, it requires  $2^{o(|V(G)|/\Delta)}=2^{o(N)}$ time since the reduction in \Cref{eq:sat_to_indepset} gives $|V(G)|=N\Delta$; this breaks the ETH. 

%
%

\subsection{Step I: SAT $\rightarrow$ Independent Set (\Cref{eq:sat_to_indepset})}\label{sec:intro:sat_indep}

Our reduction from \sat to the maximum independent set problem is simply a sequence of standard reductions (e.g., from \cite{FGLSS96,Trevisan01}) as follows: 
%
%
We start from \sat instance $\psi$, which is NP-complete but has no hardness gap.
We then construct a PCP for $\psi$, which thus outputs a CSP (Constraint Satisfaction Problem) instance $\phi_1$ with some small gap.
(Note that one may think of PCP as a reduction from a SAT instance with no gap to a CSP instance with some gap.)
Then we apply a {\em gap-amplification and sparsification scheme} to boost up the hardness-gap to (some value) $k$ and reduce the number of clauses to roughly the same as the number of variables.
So, now, we obtain another CSP instance $\phi_2$ from this step.
Finally, we apply the {\em FGLSS reduction} of Feige, Goldwasser, Lov\'asz, Safra and Szegedy \cite{FGLSS96} to obtain an instance graph $G$ of the maximum independent set problem with hardness-gap $k$.
%
%
%
These reductions give a $|V(G)|$-hardness for the independent set problem. To obtain a hardness for on $\Delta$-bounded-degree graphs, we follow Trevisan \cite{Trevisan01} who showed that by modifying the FGLSS reduction with {\em disperser replacement}, we can obtain a graph with maximum degree $\Delta\approx k$, thus implying a $\Delta$-hardness. 
Below we conclude this reduction. 

{\footnotesize
\begin{align}
\underbrace{\mbox{$N$-variable \sat $\psi$}}_{\mbox{(no hardness-gap)}}~
&\label{eq:pcp}\xRightarrow[\hspace*{5.2cm}]{\mbox{PCP}}
\quad\underbrace{\mbox{$N$-variable CSP $\phi_1$}}_{\mbox{(hardness-gap = $\gamma$)}}\\
&\label{eq:gap amplification}\xRightarrow[\hspace*{5.2cm}]{\mbox{gap amplification + sparsification}}
\quad\underbrace{\mbox{$N$-variable $(kN)$-clause CSP $\phi_2$}}_{\mbox{(hardness-gap = $k$ for all $k \geq \gamma$)}}\\
&\label{eq:fglss_disperser}\xRightarrow[\hspace*{5.2cm}]{\mbox{FGLSS + disperser replacement}}
\quad\underbrace{\mbox{Indep. Set on $G$ with $|V(G)|= N \Delta$ and $\Delta=k$}}_{\mbox{(hardness-gap = $\Delta$ = $k$)}}
\end{align}
}

Note that all tools that we use in the above are already used in the literature earlier. The hardness proof of the maximum independent set problem using gap amplification and sparsification before applying the FLGSS reduction (with no disperser replacement) were used in, e.g., \cite{BellareS94}. The hardness using FGLSS and disperser replacement (but no gap amplification and sparsification) was used in \cite{Trevisan01}. The main work in Step~I is showing how to combine all these tools together and how to tune parameters appropriately to get the desired subexponential-time hardness of the maximum independent set problem. 
Another tool that has not been used before in the context of the maximum independent set problem, which is important for our proof, is the nearly-linear size PCP of Moshkovitz and Raz \cite{MR10} and the PCP with low free-bit complexity of Samorodnitsky and Trevisan \cite{ST00}. Our reduction in a more detailed form is shown below. 

{\footnotesize
\begin{align}
\underbrace{\mbox{$N$-variable \sat $\psi$}}_{\mbox{(no hardness-gap)}}~
&\xRightarrow{\mbox{\tiny PCP (Moshkovitz-Raz + Samorodnitsky-Trevisan)}}~
\underbrace{\mbox{$N$-variable CSP $\phi_1$, acc. conf. $w_1=\gamma^{o(1)}$}}_{\mbox{(hardness-gap = $\gamma$)}}
\tag{\ref{eq:pcp}$^*$}
\label{eq:pcp refined}\\
 &\xRightarrow{\mbox{{\tiny gap amplification + sparsification}}}
 ~\underbrace{\mbox{$N$-variable $(kN)$-clauses CSP $\phi_2$, acc. conf. $w_2=k^{o(1)}$}}_{\mbox{(hardness-gap = $k = \gamma^t$ for any $t$)}}
 \tag{\ref{eq:gap amplification}$^*$}
 \label{eq:gap amplification refined} \\
  ~~&\xRightarrow[\hspace*{1cm}]{\mbox{\tiny FGLSS + disperser}}~~
  \underbrace{\mbox{Indep. Set on $G$: $\Delta=k$ and $|V(G)|=(kN)w_2$}}_{\mbox{(hardness-gap = $k$)}}
  \tag{\ref{eq:fglss_disperser}$^*$}
  \label{eq:fglss_disperser_refined}
\end{align}
}

In the above, ``acc. conf.'' is an abbreviation for the maximum number of accepting configurations. We note that among many parameters in the reduction above, two parameters that play an important role later are $\gamma$ and $t$. Below we explain the above equations. Further detail of \Cref{eq:pcp refined,eq:gap amplification refined} can be found in \Cref{sec: PCP} (\Cref{eq:pcp refined,eq:gap amplification refined} together give a nearly-linear size sparse PCP with small free-bit complexity and large degree). The detail of \Cref{eq:fglss_disperser_refined} can be found in \Cref{sec: hardness of induced matching}.

%
\paragraph{\Cref{eq:fglss_disperser_refined}}
To expand the detail of \Cref{eq:fglss_disperser} to \Cref{eq:fglss_disperser_refined}, we need to introduce another parameter $w_2$ denoting the {\em maximum number of satisfying assignments for each clause} (known as the {\em maximum number of accepting configurations}). It is quite well-known that if we start with an $N$-variable $M$-clause CSP instance $\phi_2$ with maximum number of accepting configurations $w_2$, we will obtain a graph $G$ with $|V(G)|=Mw_2$ and properties as in \Cref{eq:fglss_disperser} after applying the FGLSS reduction and the disperser replacement (for more detail, see \Cref{sec:FGLSS detail}). 
Since $M = N k$, it follows that to obtain a reduction as in \Cref{eq:fglss_disperser}, we need $w_2$ to be small; for an intuition, it is sufficient to imagine that $w_2=k^{o(1)}$. This makes $|V(G)|\approx (Nk)k^{o(1)}\approx Nk$ (note that we are being imprecise in our calculation). This leads to the refinement of \Cref{eq:fglss_disperser} as in \Cref{eq:fglss_disperser_refined}. 
%
%
%
Our next job is to obtain $\phi_2$ with the properties as in \Cref{eq:fglss_disperser_refined}; i.e., for any $k\geq \gamma$, it has $N$ variables, $kN$ clauses, and maximum number of accepting configurations $w_2=k^{o(1)}$. 

\paragraph{\Cref{eq:gap amplification refined}} We have not yet stated the precise relationship between $\gamma$ and $k$ in \Cref{eq:gap amplification} and also how to obtain $w_2=k^{o(1)}$. To do this, we have to introduce yet another parameter $t$ which we will call {\em gap amplification parameter}. By using the standard gap amplification technique with parameter $t$, we can transform a CSP instance $\phi_1$ with hardness gap $\gamma$ to an instance $\phi_2$ with hardness gap $k=\gamma^t$ with the same number of variables. This gap amplification process increases the number of clauses exponentially in $t$ but we can use the standard sparsification technique to reduce the number of clauses to $\gamma^t N$. Moreover, if the maximum number of accepting configurations in $\phi_1$ is $w_1$, then we will get $w_2=w_1^t$ as the maximum number of accepting configurations in $\phi_2$. We will show that we can obtain $\phi_1$ with $w_1=\gamma^{o(1)}$, which implies that $w_2=\gamma^{o(t)}=k^{o(1)}$. This leads to the refinement of \Cref{eq:gap amplification} as in \Cref{eq:gap amplification refined}.
%

\paragraph{\Cref{eq:pcp refined}}
Our last job is to obtain $\phi_1$ with the above properties; i.e., it has $N$ variables and the maximum number of accepting configurations is $w_1=\gamma^{o(1)}$ for the hardness gap of $\gamma$. 
%
%
Once we arrive at this step, the combination of nearly-linear size PCP of Moshkovitz and Raz \cite{MR10} and the PCP with low free-bit complexity of Samorodnitsky and Trevisan \cite{ST00}, as stated in~\cite[Corollary 14]{MR10}, gives us exactly what we need. 
%
%
This result states that any \sat formula of size $N$ can be turned into an $N^{1+o(1)}$-variable CSP with hardness gap $\gamma$ and $w_1= \gamma^{o(1)}$. This gives a refined version of \Cref{eq:pcp} as in \Cref{eq:pcp refined}.

\subsection{Step II: independent set $\rightarrow$ induced matching (\Cref{eq:indepset_to_induce})} \label{sec:intro:indep_induce}

We next reduce from the independent set problem to the bipartite induced matching problem. The key idea of this reduction is proving a new property of the disperser. The main result is that if we construct a hardness instance of the maximum independent set problem using disperser replacement (as in \cite{Trevisan01} and in this paper, as we outlined in the previous section), then the hardness of the maximum independent set and the maximum bipartite induced matching problems are essentially the same.

\begin{theorem}[Informal]\label{thm:indep to induce informal}
Let $G$ be a graph (not necessarily bipartite) constructed by disperser replacement. Then, there is a bipartite graph $H$ of roughly the same size as $G$ such that $\induce{H}\approx \alpha(G)$. 
\end{theorem}

The graph $H$ in the above theorem is actually obtained by a variant of an operation, called {\em bipartite double cover}. 
A bipartite double cover of a graph $G$, denoted by $\bipm[G]$, is a bipartite graph $H=(X, Y, E)$ where $X$ and $Y$ are copies of $V(G)$, and any vertices $x\in X$ and $y\in Y$ are adjacent in $H$ if and only if $x$ and $y$ are adjacent in $G$. 
In our variant, we also have an edge joining two vertices $x,y$ that come from the same vertex in $G$. 
Throughout, we shall abuse bipartite double cover to also mean its variant.
This operation is a natural transformation that is frequently used in transforming a graph $G$ into a bipartite graph. Many of its properties have been studied, mostly in graph theory (e.g., \cite{Brualdi1980bigraphs,Sampathkumar1975tensor,Maruvsivc2008,Imrich2008multiple}).
In this paper, we show a new property of this transformation when applied to a disperser as follows.

\subsubsection{New Property of an Old Disperser}\label{sec:disperser overview}
%
A {\em disperser} is a bipartite graph $G=(X,Y,E)$ with a certain ``expanding property'' in the sense that it has a {\em small balanced bipartite independent set} as stated informally below (the formal definition can be found in \Cref{def:disperser formal}):  
\begin{definition}[Disperser (informal)]
Let $G= (X, Y, E)$ be a disperser with some parameter $k>0$.
Then $G$ has {\em no} independent set $S$ such that $|S\cap X|=|S\cap Y|\geq k$.  
\end{definition}
In this case, we say that the balanced bipartite independent set number of $G$, denoted by $\BBIS(G)$ is at most $k$. (For an intuition, think of $k$ as $k=\delta |X|$ for some small value $\delta$.) The disperser plays an important role in proving hardness of approximation; see, e.g., the use of disperser in Trevisan's construction in \cite{Trevisan01} for the hardness of approximating the maximum independent set problem on a bounded-degree graph. 
%
So, it is natural to study the properties of dispersers when considering the hardness of approximating the maximum independent set and related problems (the maximum induced matching problem, in particular).
%
%
In this paper, we show that if $G$ is a disperser, then there is a tight connection between its balanced bipartite independence number $\BBIS(G)$, its induced matching number, and the induced matching number of its bipartite double cover $\induce{\bipm[G]}$:
\begin{lemma}[Disperser Lemma (Informal)]\label{lmm:disperser informal}
If $G=(X, Y, E)$ is disperser, then 
$$\induce{\bipm[G]} = O(\induce{G}) = O(\BBIS(G)).$$
\end{lemma}
(The formal version of the above lemma can be found in \Cref{lem: properties of dispersers}.)
In fact, \Cref{lmm:disperser informal} holds for any bipartite graph $G$, but we only need to apply it to dispersers. 
This lemma is crucial in our construction.
Its proof is so simple that we can sketch it here.
First, we show that $\induce{G} = O(\BBIS(G))$.
Let $M = \{x_1y_1, \ldots, x_ty_t\}$ be any induced matching of size $t$ in $G$, where $x_i\in X$ and $y_i\in Y$ for all $i$.
Observe that the set $S=\{x_1, \ldots x_{\lfloor t/2\rfloor}\}\cup\{y_{\lfloor t/2\rfloor+1}, \ldots, y_t\}$ is a balanced bipartite independent set.
So, it follows that $\BBIS(G)\geq \lfloor \induce{G}/2\rfloor$ and thus $\induce{G} = O(\BBIS(G))$ as desired.
Observe that our proof simply exploits the fact that $G$ is bipartite.

Next, we prove that $\induce{\bipm[G]} = O(\induce{G})$.
As the equation looks natural, one might wonder if this holds for {\em any graph}.
Unfortunately, we have to answer this question negatively since there is a simple counter example showing that this claim does not hold for a general graph. 
(Consider, for example, a graph $H$ with a vertex set $\{x_1, \ldots, x_t\}\cup \{y_1, \ldots, y_t\}$ where edges are in the form $x_iy_i$, $x_ix_j$ and $y_iy_j$, for all $i$ and $j$. In other words, the graph $H$ consists of two equal-sized cliques that are connected by a perfect matching. It can be seen that $\induce{H}=2$ whereas $\induce{\bipm[G]}\geq t/2$.) 
Our proof again exploits the fact that the graph $G$ is bipartite. 
Consider an induced matching $M$ in $\bipm[G]$. Let $U'$ and $U''$ be the two bipartitions of $\bipm[G]$. 
For any vertex $x$ in the bipartition $X$ of $G$, let us write its copy in $U'$ and $U''$ as $x'$ and $x''$, respectively. We do the same thing for every vertex $y\in Y$. 
Observe that edges in $\bipm[G]$ must be in the form $x'y''$ or $x''y'$, for some $x\in X$ and $y\in Y$. 
%
It follows that $M$ (an induced matching in $\bipm[G]$) must be in the form $\{x'_1y''_1, \ldots, x'_ty''_t\} \cup \{x''_{t+1}y'_{t+1}, \ldots, x''_{t'}y'_{t'}\}$ where, for all $i$, $x_i\in X$, $y_i\in Y$. We can use $M$ to construct two induced matchings in $G$: $M_1=\{x_1y_1, \ldots, x_ty_t\}$ and $M_2 = \{x_{t+1}y_{t+1}, \ldots, x_{t'}y_{t'}\}$ (it is not hard to show that $M_1$ and $M_2$ are induced matchings). Thus,  $\induce{G}\geq \induce{\bipm[G]|}/2$, implying that $\induce{\bipm[G]} = O(\induce{G})$ as desired.

The property of dispersers in \Cref{lmm:disperser informal}, when plugged into the construction of Trevisan \cite{Trevisan01}, immediately implies \Cref{thm:indep to induce informal}. 
Combining \Cref{thm:indep to induce informal} with the hardness of the maximum independent set problem, we immediately obtain the hardness of the maximum induced matching problem in a bipartite graph.\danupon{BTW, we should make it clear that we always talk about bipartite graphs}  Further detail of this step can be found in \Cref{sec: hardness of induced matching}


\subsection{Step III: Induced matching $\rightarrow$ Pricing (\Cref{eq:indepset_to_pricing})}\label{sec:intro:induce_pricing}

In the last step of our reduction, we reduce from the bipartite induced matching problem to the pricing problem. Given a bipartite graph $G=(U,W,E)$, we convert it to an instance of the pricing problem by thinking of the left vertices $U$ as consumers having different budgets and the right vertices $W$ as items. The main idea is that the price of each item will tell us which consumer it should be matched to (to get a solution for the bipartite induced matching problem); i.e., if an item $I$ has price $p$, then it should be matched to the consumer of budget roughly $p$. To make this idea work, we have to make consumers' budgets {\em geometrically} different, as well as the number of consumers. To be precise, we need to convert each vertex $u\in U$ into some $2^i$ consumers in such a way that any two vertices $u$ and $u'$ sharing the same neighbor $v$ must be converted to different numbers of consumers. This leads to an exponential blowup. We need a blowup of $2^{O(\Delta)}$ intuitively because each vertex $v\in V$ has degree $\Delta$. This is how the exponential blowup comes into picture.  
Further detail of this step can be found in \Cref{sec:hardness_pricing}.

%% file: prelim.tex
\section{Preliminaries}
\label{sec:prelim}

We use standard graph terminologies as in ~\cite{DiestelBook}.
We denote any graph by $G=(V,E)$.
When we consider more than one graph, we may denote the set of vertices and edges of a graph $G$ by $V(G)$ and $E(G)$, respectively.
A set of vertices $S\subseteq V$ is {\em independent} (or {\em stable}) in $G$ if and only if $G$ has no edge joining a pair of vertices $u,v\in V$. 
A set of edges $M\subseteq E$ is a {\em matching} in $G$ if and only if no two edges of $M$ share an end-vertex, and a matching $M$ is an {\em induced matching} in $G$ is the subgraph of $G$ induced by $M$ is exactly $M$. 
That is, $M$ is an induced matching in $G$ if and only if, for every pair of edges $uv, ab\in M$, $G$ has none of the edges $ua,ub,va,vb$.

\paragraph{Semi-Induced Matching:}
To prove the hardness of pricing problems, we need a slight variant of the maximum induced matching problem, called the maximum semi-induced matching problem.   
Given a {\em permutation} (a.k.a, a {\em total order}) $\sigma$ of $V$, a set of edges $M\subseteq E$ is a {\em $\sigma$-semi-induced matching} in
$G$ if and only if, for every pair of edges $uv, ab \in M$ such
that $\sigma(u)<\sigma(a)$, $G$ has none of the edges $ua, ub$.
Given any graph $G$ and a total order $\sigma$, we use the notation $\sinducesigma{\sigma}{G}$ to denote the size of a maximum $\sigma$-semi-induced matching in $G$, and let $\sinduce{G} = \max_{\sigma} \sinducesigma{\sigma}{G}$.  
Notice that for any $\sigma$, if $M$ is an induced matching in $G$, then $M$ is also a $\sigma$-semi-induced matching in $G$, so we must have $\induce{G} \leq \sinducesigma{\sigma}{G} \leq \sinduce{G}$.  
In the maximum semi-induced matching problem, our goal is to compute $\sinduce{G}$.  

Our hardness proof of the maximum bipartite induced matching problem will, in fact, show a stronger property than just bounding the size of a maximum induced matching. 
That is, in the completeness case, the reduction guarantees that $\induce{G} \geq c$ while in the soundness case, it gives $\sinduce{G} \leq s$ (where $s$ and $c$ are soundness and completeness parameters, respectively).
Notice that $\sinduce{G} \leq s$ implies $\induce{G} \leq s$, so this stronger property implies the $(c/s)$-hardness of approximating the maximum bipartite induced matching problem as a consequence. 
Bounding the size of a maximum semi-induced matching in the soundness case seems to be necessary for us here in order to prove the hardness of the $k$-hypergraph pricing problem.

\paragraph{The Pricing Problems:}
In an equivalent formulation of the $k$-hypergraph pricing problem that we will be using throughout the paper, the pricing instance is given by two sets $(\cset, \iset)$ where $\cset$ and $\iset$ are the sets of consumers and items, respectively. 
Each consumer $c \in \cset$ is associated with a budget $B_c$ and an item set $S_c \subseteq \iset$. 
We have an additional constraint that $|S_c| = k$.  
It is easy to see that this formulation is equivalent to the hypergraph formulation, i.e., each vertex corresponds to an item and each edge corresponds to a consumer, and the additional constraint $|S_c| = k$ ensures that the size of each hyperedge is $k$. 
For our purpose, this formulation has an advantage over the hypergraph formulation since we will be dealing with several graph problems and connections between them, which can easily create confusion between a number of graphs that come up in the same context.   
\paragraph{Constraint Satisfaction Problems:} 
One of the most fundamental problems in theoretical computer science is
$k$-SAT, where we are given a CNF formula $\phi$, and the goal is to
decide whether there is an assignment to boolean variables of $\phi$
that satisfies all the formula. 
In the maximization version of $k$-SAT, the goal is to find an
assignment that maximizes the number of clauses of $\phi$ satisfied. 

The {\em $k$-constraints satisfaction problem ($k$-CSP)} is a
generalization of $k$-SAT, in which each clause is a (boolean) function
$\Pi_j$ on $k$ (boolean) variables, and the goal of $k$-CSP is to find
an assignment to variables that satisfies as many clauses as possible.
That is, the goal is to find an assignment $f$ such that 
$\Pi_j(f_j)=1$ for all clause $\Pi_j$, where $f_j$ is a 
{\em partial assignment} restricted to only variables that appear in $\Pi_j$. 
We use the term {\em assignment of a clause} $\Pi_j$ to mean  
a partial assignment restricted to variables in $\Pi_j$.
For example, suppose $\Pi_j$ consists of variables $x_1,x_2,x_3$; 
then $(1,0,1)$ is an assignment of $\Pi_j$ where $x_1=1$, $x_2=0$ and
$x_3=1$. 
An assignment $f_j$ of $\Pi_j$ is a {\em satisfying assignment} 
for the clause $\Pi_j$ if and only if $\Pi_j(f_j)=1$, 
i.e., $f_j$ satisfies $\Pi_j$.

\paragraph{Exponential Time Hypothesis:} The complexity assumption that has received more attention recently is the 
{\em Exponential Time Hypothesis} (ETH), which rules out the existence of
subexponential-time algorithms that decide $k$-SAT. 
The formal statement of this conjecture is as follows.

\begin{hypothesis}[Exponential-Time Hypothesis (ETH)]
\label{hypo:ETH}
For any integer $k\geq 3$, there is a constant $0<q_0(k)<1$ such that
there is no algorithm with a running time of $2^{q N}$, for all
$q<q_0(k)$, that solves $k$-SAT where $N$ is the size of the instance.
In particular, there is no subexponential-time algorithm that solves
$3$-SAT. 
\end{hypothesis}

Indeed, the ETH was first stated in terms of the number of variables. 
Impagliazzo, Paturi and Zane \cite{IPZ01} showed that the statement
is equivalent for all the parameters, i.e., 
$N$ in the statement can be the number of variables, the number of
clauses or the size of the instance.
For our purpose, we state the theorem in terms of the instance size. For more discussion related to the ETH, we refer readers to a comprehensive survey by Lokshtanov et~al.~\cite{LokshtanovMS11} and references therein.

%% file: pcp.tex

\section{Nearly-linear size sparse PCP with small free-bit complexity and large degree}
\label{sec: PCP} 

This section implements the reductions in~\Cref{eq:pcp refined} and \Cref{eq:gap amplification refined} as outlined in the overview section. 
%
%
Informally, given a SAT formula, we need to construct a CSP with the following property: (1) We need it to be {\em sparse}, i.e., the number of clauses is small compared to the number of variable. (2) The {\em amortized free-bit complexity}, i.e., the number of satisfying assignments of each clause, can be made arbitrarily small. (3) The {\em degree}, i.e., the number of occurrences of each variable, is large. The last requirement is one of the main reasons that we have to slightly modify a PCP. 
Moreover, to be able to use the ETH, we need the size, i.e., the number of variables, to be nearly-linear, so we start from the nearly linear sized PCP of Moshkovitz and Raz~\cite{MR10}. 
The following theorem summarizes the properties of our PCP.

\begin{theorem}[Nearly-linear size sparse PCP with small free-bit complexity and large degree]
\label{thm: CSP} 
Let $k, t$ be parameters and $\epsilon = 1/k$. Also, let $\delta>0$ be any parameter. 
There is a randomized polynomial-time algorithm that transforms a 3SAT formula of size $N$ to a $(tq)$-CSP formula $\phi$, where $q = k^2 + 2k$, that satisfies the following properties: 
\begin{itemize} 
\item (Small Number of Variables and Clauses) The number of variables is at most $q N^{1+\epsilon}$\danupon{LATER: Before it was $q N^{1+\epsilon+\delta}$. We should check whether we have to make changes in later sections due to this change.}, and the number of clauses is $M = 100q2^{t(k^2+1)}N^{1+\epsilon+\delta}$. 

\item (Big Gap between Completeness and Soundness) The value of the \yi is at least $c=1/2^{t+1}$, and the value of \ni is $s=2^{-t(k^2+2)}$.

\item (Free-Bit Complexity) For each clause in $\phi$, the number of satisfying assignments for such clause is $w=2^{2kt}$. Moreover, for each variable $x_j$ that appears in a clause, the number of satisfying assignments for which $x_j = 0$ is equal to the number of satisfying assignments for which $x_j = 1$.  

\item (Large Degree) For each variable $x_j$, the total number of clauses in which $x_j$ appears is $M_j \geq N^{\delta}2^{t(k^2+1)}$.\danupon{This is what is different from other PCPs}  

\end{itemize}  
\end{theorem} 

The outline of our proof is as follows.
First, we take a PCP with nearly-linear size and optimal amortized free-bit complexity.  
We slightly modify PCP to satisfy our desired property and
then transform the PCP into a CSP instance.
The CSP instance that we obtain at this point has small hardness gap (i.e., the ratio between completeness and soundness parameters), 
so we apply a standard {\em gap-amplification scheme} (see~\cite{Zuckerman96}) to 
amplify this gap and obtain a final CSP instance.

\subsection{PCP Construction}
\label{sec:PCP construction}


A {\em probabilistically checkable proof} system (PCP) is a characterization of an NP problem.
A PCP system for a language $L$ of size $N$ consists of a (randomized)
verifier $V$ on an input $\phi$ that has an oracle access to a proof 
string $\Pi$ given by a prover.
To decide whether to ``accept'' or ``reject'' the proof,
the verifier flips a coin a certain number of times, forming a random string $r$,
and then queries $q$ positions of the proof corresponding to $r$. 
(These $q$ positions depends on the random string $r$.)
We call the values read from the proof string {\em answers}.
Given a random string $r$, we denote by $b_1(r),\ldots, b_q(r)$ 
the indices of the proof bits read by the verifier. 
A {\em configuration} is a $(q+1)$-tuple $(r, a_1,\ldots, a_q)$ 
where $a_i$ is the value of the position $b_1(r)$ of the proof bit.
We say that a configuration $(r,a_1,\ldots, a_q)$ is {\em accepting} if the random string $r$ 
of the verifier and the answers $a_1,\ldots, a_q$ from the prover cause the verifier 
to accept the proof.
A subset of $f$ queries are {\em free} if, for any possible answers
to these queries, there is a unique answer to each of the other
queries that would make the verifier accepts.


Our starting point in constructing the desired CSP 
is the following PCP characterization of NP by Moshkovitz and Raz~\cite[Corollary 14]{MR10}.


\begin{theorem}[\cite{MR10} + \cite{ST00}]
\label{thm: MR}
For any sufficiently large constant $k \geq \omega(1)$,
3SAT on input of size $N$ has a verifier that uses
$(1+ \epsilon) \log{N}$ random bits to pick $q = k^2 + 2k$ queries 
to a binary proof, such that only $2k$ of the queries are free, 
i.e.,  for each random string, there are $2^{2k}$ possible
satisfying assignments of the queried bits in the proof.
The verifier has completeness $1-\epsilon$ and
soundness error at most $2^{-k^2 +1}$.
Moreover, the acceptance predicate is linear.
\end{theorem}
In fact, the main result in \cite{MR10} is the 2-query projection PCP
(a.k.a. the label~cover problem) with sub-constant error and nearly-linear size.
In \cite{ST00}, Samorodnitsky and Trevisan constructed a PCP with
optimal amortized free-bit complexity via linearity testing and
used a 2-query projection PCP as an outer PCP.
For simplicity of parameters and since we only need the completeness of $1/2$ in our construction, we replace the completeness of $1-\epsilon$ by $1/2$ from now on.\parinya{Added this sentence.} 

Since the acceptance predicate of the verifier is linear, 
we can assert that, for each random string $r$ and each integer $j: 1 \leq j \leq q$, 
the number of accepting configurations for which $\Pi_{b_j(r)} = 0$ 
is equal to the number of accepting configurations for which $\Pi_{b_j(r)} =1$.  

Now we describe the (slight) modification of PCP. 
We want each proof bit to be participated in at least $N^\delta$ random strings in order to ensure sufficient success probability in later steps. 
Therefore, we modify the PCP by allowing the verifier to flip $\delta \log N$ extra random coins, so now the number of random bits needed is $(1+\epsilon+ \delta) \log N$. 
The verifier does not perform any extra computation based on these new random coins.
\bundit{Just added the following -- Aug 10, 2013}
That is, the positions of the proof that the verifier reads depend only on the first $(1+\epsilon)\log N$ random bits. This guarantees that the proof degree is at least $N^\delta$ while the completeness and soundness are preserved.
(In the context of CSP, this is equivalent to making $2^{\delta \log N}$ copies of each clause.)
\bundit{Just added - Aug 11, 2013}
This step is required to ensure that the PCP has a proof degree polynomial on $N$ (which is later required in the construction of dispersers.) When $t=\Omega(\log N)$, we can skip this step because the proof degree will be $\poly(N)$ automatically after the gap amplification step.
%
To be precise, given a PCP verifier $V$, we construct a modified PCP verifier $V'$
as follows.\bundit{Should we remove the detail here to make it brief? In fact, it should be in appendix.}
The verifier $V'$ uses the same proof $\Pi$ as that of $V$, and $V'$ has 
the same parameters as $V$ except for the size of a random string.
Each time, $V'$ flips $B=(1+\epsilon+ \delta) \log N$ coins to 
generate a random string $r'=(x_1,\ldots,x_B)$, where $x_i\in\set{0,1}$.
Then $V'$ simulates the verifier $V$ and feeds $V$ a random string
$r=(x_1,\ldots,x_{(1+\epsilon) \log N})$. 
So, the positions of the proof $\Pi$ that $V$ reads depend only on the 
first $(1+\epsilon) \log {n}$ random bits generated by $V'$.
The verifier $V'$ accepts the proof if (the simulation of) $V$ accepts
the proof. 
It can be seen that the verifier $V'$ has all the properties of $V$, but the 
{\em proof degree}, i.e., the number of random strings that cause 
$V'$ to read each position $i$ of the proof, is now at least $N^\delta$.
Hence, we have the following theorem.

\begin{corollary}
\label{cor: modified-PCP}
For any sufficiently large constant $k \geq \omega(1)$,
3SAT on input of size $N$ has a verifier that uses
$(1+ \epsilon +\delta) \log{N}$ random bits to pick $q = k^2 + 2k$ queries 
to a binary proof, such that only $2k$ of the queries are free, 
i.e.,  for each random string, there are $2^{2k}$ possible
satisfying assignments of the queried bits in the proof.
The verifier has completeness $1-\epsilon$ and
soundness error at most $2^{-k^2 +1}$, and the acceptance predicate is linear.
Moreover, the proof degree is at least $N^\delta$, i.e.,
for each position $i$ of the proof $\Pi$, 
there are at least $N^\delta$ random strings that cause 
the PCP verifier to read $\Pi(i)$. 
\end{corollary}
\begin{proof}
We start from the PCP as in Theorem~\ref{thm: MR} and modify the PCP
as in the above discussion. 
Since all the computations are done by (the simulation of) the verifier $V$,
it is easy to see the properties of $V$ and $V'$ 
(i.e., the size of the proof, the number of queries, the number of free queries and 
the linearity of the acceptance predicate) are the same. 
The proof degree of the verifier $V'$ follows by the construction.
That is, for each random string $r=(x_1,\ldots,x_{(1+\epsilon)\log N})$ of $V$,
there are $N^\delta$ random strings 
$r'=(x_1,\ldots,x_{(1+\epsilon)\log N},y_{1},\ldots,y_{\delta\log N})$
(generated by $V'$) that causes $V'$ to feed the random string $r$ to $V$.
Thus, if $r$ causes $V$ to read the position $i$ of the proof, then there are 
at least $N^\delta$ random strings of $V'$ that causes $V'$ (indeed, $V$) to read the position $i$ of the proof. 
In other words, $V'$ has proof degree at least $N^\delta$.

Finally, as $V'$ has more random bits than $V$, we need to verify that the completeness and soundness of these two PCPs are the same. Let us check this quickly. Fix a proof $\Pi$. For each random string $r$ of $V$, there are exactly $N^\delta$ random strings $r'$ of $V'$ that contains $r$ as the first $(1+\epsilon) \log{n}$ substring, and the acceptance and rejection of the proof $\Pi$ depends only on $r$. Thus, the fraction of random strings that cause $V'$ to accept or reject $\Pi$ is the same as that of $V$. In other words, the modification of the PCP preserves both the completeness and soundness. This proves the corollary.
\end{proof}

From this PCP, we create a CSP formula $\phi$ in a natural way: For each proof bit $\Pi_b$, we have a corresponding variable $x_b$ that represents the proof bit. 
For each random string $r$, we have a clause that involves variables $x_{b_1(r)},\ldots, x_{b_q(r)}$, and this clause is satisfied by an assignment $(a_1,\ldots, a_q)$ if and only if $(r, a_1,\ldots, a_q)$ is an accepting configuration.
The followings summarizes the properties of the resulting CSP: 

\begin{itemize}
\item ({\sc Size:}) The number of clauses and variables is at most $m=q N^{1+\epsilon+\delta}$. 
This is simply because the length of random strings is $(1+\epsilon +\delta) \log N$.  

\item (\yi:) If the input 3-SAT instance is satisfiable, then there is an assignment that satisfies $1/2$ fraction of clauses of $\phi$.

\item (\ni:) If the input 3-SAT instance is not satisfiable, then any assignment satisfies at most $2^{-k^2+1}$ fraction of the clauses of $\phi$.

\item ({\sc Balancedness:}) For each clause, the number of assignments that satisfy such clause is at most $2^{2k}$, and for each variable $x_j$ appearing in such clause, the number of satisfying assignments for which $x_j=1$ equals that for which $x_j=0$.  

\item ({\sc Degree of CSP:})For any variable $x_j$, the total number of clauses that contain $x_j$ is $m_j \geq N^{\delta}$.  
\end{itemize}   

\paragraph{Gap Amplification:}
Now, we have already generated a CSP with relatively large gap between \yi and \ni.  
In this step, we amplify the gap between the \yi and \ni further by a gap amplification scheme (see \cite{Zuckerman96}), which will imply the theorem.  

%
Let $M= 100 q 2^{t(k^2+1) }N^{1+\epsilon+\delta}$.
We construct from $\phi$ a new CSP instance $\phi'$ on $M$ clauses.
For $i=1,2,\ldots,M$, we create a clause $C'_i$ of $\phi'$ by independently and uniformly at random choosing $t$ clauses from the formula $\phi$ and join them by the operation ``AND''. 
It can be seen that the number of variables of $\phi'$ is the same as that of $\phi$, and the total number of clauses is exactly $M$. 
It remains to prove the other properties of the CSP instance $\phi'$ using probabilistic arguments.

\begin{itemize}

\item {\sc Completeness:}
%
First, we prove the completeness.
In the \yi, there is an (optimal) assignment $\sigma$ of $\phi$ that satisfies $1/2$ fraction of the clauses in $\phi$. 
We will show that with high probability $\sigma$ satisfies at least $1/2^{t+1}$ fraction of the clauses of $\phi'$.
(Note that $\phi$ and $\phi'$ have the same set of variables.)
Consider the random process that constructs $\phi'$.
The probability that each randomly generated clause of $\phi'$ is satisfied by $\sigma$ (i.e., the probability that all the $t$ clauses are satisfied by $\sigma$) is at least $1/2^t$ .
So, the expected number of clauses satisfied by $\sigma$ is $2^{-t} M$.
Since $2^{-t} M \geq 100 q N^{1+\delta+\epsilon}$ (by our choice of $M$), we can apply Chernoff's bound to show that the probability that $\sigma$ satisfies less than $1/2^{t+1}$ fraction of clauses in $\phi'$ is less than $1/N$.

\item {\sc Soundness:}
Now, we prove the soundness.
In the \ni, any assignment $\sigma$ satisfies at most $2^{-k^2+1}$ fraction of the clauses of $\phi$.
We will show that $\sigma$ satisfies at most $2^{-t(k^2-2)}$ fraction of the clauses of $\phi'$ with high probability.
So, for any assignment $\sigma$, the expected number of clauses in $\phi'$ satisfied by $\sigma$ is at most $2^{-t(k^2-1)}M \leq 100 q N^{1+\delta +\epsilon}$. 
By Chernoff's bound, the probability that such assignment satisfies more than $2^{-t(k^2-2)} M$ clauses is at most $2^{-10qN^{1+\delta+\epsilon}}$.
Since there are at most $2^{qN^{1+\delta+\epsilon}}$ such assignments (because the number of variables is $q N^{1+\delta+\epsilon}$), we have by union bound the claimed soundness. 

\item {\sc Balanceness:}
The property in the third bullet holds trivially because each clause of $\phi'$ is constructed from the ``AND'' of $t$ clauses of $\phi$. 

\item {\sc Degree of CSP:}
Finally, to prove the fourth bullet, we fix a variable $x_j$. 
The probability that each random clause contains a variable $x_j$ is at least $m_j/m$.
So, the expected number of clauses that contain $x_j$ is at least 
\[
  \frac{m_j}{m} M \geq \frac{m_j}{q N^{1+\epsilon+\delta}} 100 q 2^{t(k^2+1)} N^{1+\delta+\epsilon} 
                   \geq 100 N^{\delta} 2^{t(k^2+1)} 
\] 
By applying Chernoff's bound, the probability that we have less than $N^{\delta} 2^{t(k^2+1)}$  clauses containing $x_j$ is at most $1/N$. 
\end{itemize}

This completes the proof of Theorem~\ref{thm: CSP}.

%% file: hardness-matching.tex
\section{Tight Hardness of Semi-Induced Matching}
\label{sec: hardness of induced matching}

In this section, we prove the (almost) tight hardness result of the
semi-induced matching problem on a $\Delta$-degree bounded bipartite graph.
What we prove is actually stronger than the hardnesses of the induced and semi-induced matching problems themselves: We show that the completeness case has a large induced matching while the soundness case has no large semi-induced matching.   
The formal statement is encapsulated in the following theorem.
\danupon{We should use at most one of $\Delta$-degree bounded or $d$-degree bounded (not both).}

\begin{restatable}[Hardness of $\Delta$-Degree Bounded Bipartite Semi-induced Matching]{theorem}{inducedmatching}
\label{thm: hardness of semi-induced matching}
Let $\epsilon> 0$ be any constant and $t>0$ be a positive integer.  
There is a randomized algorithm that transforms
a SAT formula $\phi$ of input size $N$ 
into a $\Delta$-degree bounded bipartite graph, where $\Delta=2^{t(\frac{1}{\epsilon^2}+ O(\frac{1}{\epsilon})  )}$ such that:
\begin{itemize}
  \item (\yi:) If $\phi$ is satisfiable, then $\induce{G} \geq |V(G)|/\Delta^{\epsilon}$.

  \item (\ni:) If $\phi$ is not satisfiable, then $\sinduce{G} \leq |V(G)|/\Delta^{1-\epsilon}$.
\end{itemize}
The construction size is $|V(G)| \leq  N^{1+\epsilon}\Delta^{1+\epsilon}$,
and the running time is $\poly(N,\Delta)$. 
Moreover, as long as $t \leq 5 \epsilon^2 \log N$, the reduction is guaranteed to be successful with high probability. 
\end{restatable}
%

\begin{restatable}[Hardness of $d$-Degree-Bounded Maximum Independent Set]{theorem}{hardnessindep}
\label{thm: hardness of independent set}
Let $\epsilon> 0$ be any sufficiently small constant and $t>0$ be a positive integer.
There is a randomized algorithm that transforms
a SAT formula $\phi$ on input of size $N$
into a $d$-degree-bounded graph $G$, where $d=2^{t(\frac{1}{\epsilon^2}+ O(\frac{1}{\epsilon})  )}$ such that:
\begin{itemize}
  \item (\yi:) If $\phi$ is satisfiable, then $\alpha(G) \geq |V(G)|/d^{\epsilon}$.

  \item (\ni:) If $\phi$ is not satisfiable, then $\alpha(G) \leq |V(G)|/d^{1-\epsilon}$.
\end{itemize}
The construction size is $|V(G)| = N^{1+\epsilon}d^{1+\epsilon}$,
and the running time is $\poly(N,d)$. 
Moreover, as long as $t \leq 5 \epsilon^2 \log N$, the reduction is guaranteed to be successful with high probability. 
\end{restatable}

\subsection{The Reduction}
\label{sec:reduction-semi-induced-matching}

Our reduction is precisely described as follows.
Take an instance $\phi$ of $(qt)$-CSP as in Theorem~\ref{thm: CSP}
that has $N$ variables and $M$ clauses.

\paragraph{The FGLSS Graph $\widehat{G}$ with Disperser Replacement:}
First, we construct from $\phi$ a graph $\widetilde{G}$ by
the FGLSS construction, and then the graph $\widetilde{G}$ will be transformed to graph $\widehat{G}$ by the disperser replacement step. 
For each clause $\phi_j$ of $\phi$, for each possible satisfying assignment $C$ of $\phi_j$, we create in $\widetilde{G}$ a vertex $v(j, C)$ representing the fact that ``$\phi_j$ is satisfied by assignment $C$''.
Then we create an edge $v(j, C) v(j', C')\in E(\widetilde{G})$ if there is a
{\em conflict} between partial assignments $C$ and $C'$, i.e.,
there is a variable $x_i$ appearing in clauses $\phi_j$ and $\phi_{j'}$ such that
$C$ assigns $x_i=0$ whereas $C'$ assigns $x_i=1$.
So, the total number of vertices is $|V(\widetilde{G})|=w \cdot M$.
The independence number of a graph $\widetilde{G}$ corresponds to the
number of clauses of $\phi$ that can be satisfied.
In particular, we can choose at most one vertex from each clause $\phi_j$
(otherwise, we would have a conflict between $v(j,C)$ and $v(j,C')$), and we can choose two vertices
$v(j,C),v(j',C')\in V(\widetilde{G})$ if and only if the assignment $C$ and $C'$ have no
conflict between variables.
Thus, the number of satisfiable clauses of $\phi$ is the same as the
independence number $\alpha(\widetilde{G})$.
Hence, in the \yi, we have $\alpha(\widetilde{G})\geq c \cdot M$, and in \ni, we have
$\alpha(\widetilde{G}) \leq s \cdot M$.
This gives a hard instance of the maximum independent set problem.
Notice that the degree of $\widetilde{G}$ can be very high. 
%

Next, in order to reduce the degree of $\widetilde{G}$, we follow the disperser replacement step as in \cite{Trevisan01}.
Consider an additional property of $\widetilde{G}$.
For each variable $x_i$ in $\phi$, let $O_i$ and $Z_i$ denote the set of
vertices $v(j,C)$ corresponding to the (partial) assignments
for which $x_i=1$ and $x_i=0$, respectively.
It can be deduced from Theorem~\ref{thm: CSP} 
that $|O_i| = |Z_i| =M_i/2 \geq 2^{t(k^2+1)}N^{\delta}$, for some constant $\delta>0$.

Since there is a conflict between every vertex of $O_i$ and $Z_i$,
these two sets define a complete bipartite subgraph of $\widetilde{G}$, namely
$\widetilde{G}_i=(O_i,Z_i,\widetilde{E}_i)$, where $\widetilde{E}_i=\set{uw:u\in O_i, w\in Z_i}$.
Observe that if we replace each subgraph $\widetilde{G}_i$ of $G$ by
a $d$-degree bounded bipartite graph, the degree of vertices in the resulting graph reduces to $qt d$. 
To see this, we may think of each vertex $u$ of $\widetilde{G}$ as a vector with $qt$ coordinates (since it corresponds to an assignment to some clause $\phi_j$ which has $qt$ related variables).
For each coordinate $\ell$ of $u$ corresponding to a variable $x_i$,
there are $d$ neighbors of $u$ having a conflict at coordinate $\ell$
(since the conflict forming in each coordinate are edges in
$\widetilde{G}_i$, and we replace $\widetilde{G}_i$ by a $d$-degree bounded
bipartite graph).
Thus, each vertex $u$ has at most $qtd$ neighbors.
However, as we wish to preserve the independence number of $G$, i.e.,
we want $\alpha(\widehat{G})\approx\alpha(\widetilde{G})$, we require such
a $d$-degree bounded graph to have some additional properties.
To be precise, we construct the graph $\widehat{G}$ by replacing each subgraph
$\widetilde{G}_i$ of $\widetilde{G}$ by a {\em $(d,\gamma)$-disperser}
$H_i=(O_i,Z_i,E_i)$, defined below. 

\begin{definition}[Disperser] \label{def:disperser formal}
\parinya{TODO for later: Present the dispersers and its properties first before the construction.}
A $(d,\gamma)$-disperser $H=(U',W',E')$ is a $d$-degree bounded bipartite graph on
$n'=|U'|=|W'|$ vertices such that,
for all $X\subseteq U', Y \subseteq W'$, if
$|X|, |Y| \geq \gamma n'$, then there is an edge $xy \in E'$ joining a
pair of vertices $x\in X$ and $y\in Y$.
\end{definition}

Intuitively, the important property of the disperser $H_i$ is that 
any independent set $S$ in $H_i$ cannot contain a large number of
vertices from both $O_i$ and $Z_i$;
otherwise, we would have an edge joining two vertices in $S$.

All these ideas of using disperser to ``sparsify'' the graphs were used by Trevisan in~\cite{Trevisan01} to prove the hardness of the bounded degree maximum independent set problem.
The key observation that makes this construction work for our problem is that a similar property that holds for the size of a maximum independent set also holds for the size of 
a maximum $\sigma$-semi-induced matching in $\bipp[H_i]$, i.e., 
$\bipp[H_i]$ cannot contain a large $\sigma$-semi-induced matching, for
any permutation $\sigma$.

Now, we proceed to make the intuition above precise. 
A $(d,\gamma)$-disperser can be constructed by a randomized algorithm, which is stated in the next lemma.
In the case that $d$ is constant, we may construct a $(d,\gamma)$-disperser
by a deterministic algorithm in \cite{RVW00}, which has a running time
exponential in terms of $d$. 

\begin{lemma}
\label{lemma: disperser}
For all $\gamma >0$ and sufficiently large $n$, there is a randomized
algorithm that with success probability $1- e^{-n\gamma (\log (1/\gamma) -2)}$,
outputs a $d$-regular bipartite graph
$H=(O, Z, E), |O| = |Z| = n$, where $d=(3/\gamma)\log(1/\gamma)$
such that, for all $X\subseteq Z, Y \subseteq O$,
if $|X|, |Y| \geq \gamma n$, then there is an edge $(x,y) \in E$
joining a pair of vertices $x\in X$ and $y \in Y$.
\end{lemma}

The condition that $n_i$ is sufficiently large is satisfied because
$|O_i|=|Z_i|\geq M_i\geq N^{\delta} 2^{tk^2}$ for all $i$ (since each variable $x_i$ appears in $M_i$ clauses, and for each such clause, there is at least one accepting configuration for which $x_i = 0$ and one for which $x_i = 1$.)  
\bundit{Just edited below: Aug 11, 2013}
Also, since the success probability in constructing each disperser is high (i.e., at least $2^{N^\delta}$), we can guarantee that all the dispersers are successfully constructed with high probability.
%
%
%
By setting appropriate value of $\gamma$ (which we will do later) and 
following analysis in \cite{Trevisan01}, 
we have the following completeness and soundness parameters with high
probability.

\begin{itemize}
  \item (\yi:) $\alpha(\widehat{G}) \geq 2^{-t}M$
  \item (\ni:) $\alpha(\widehat{G}) \leq 
                2^{-t(k^2+2)}M +  \gamma qt (w M)$
\end{itemize}

\paragraph{The Final Graph $G$:}
We construct the final graph $G$ by transforming $\widehat{G}$ into a bipartite
graph as follows:
first create two copies $V'$ and $V''$ of vertices of $\widehat{G}$, i.e., each vertex $u \in V(G)$ has two corresponding copies $u' \in V'$ and $u'' \in V''$;
then create an edge joining two vertices $u'\in V'$ and $w''\in V''$ if and only if
there is an edge $uw\in E(\widehat{G})$ or $u=w$.
Thus, a formal definition can be written as $G=B[\widehat{G}]= (U \cup W,E)$ where  
\begin{eqnarray*}
U&=&\set{(u,1):u\in V(\widehat{G})}, \\
W&=&\set{(w,2):w\in V(\widehat{G})} \mbox{ and } \\
E&=&E_1 \cup E_2 \mbox{ where} \\
 E_1 &=& \set{(u,1)(u,2): u \in V(\widehat{G})}, \\
 E_2 &=&  \set{(u,1)(w,2):u,w\in V(\widehat{G}) \wedge (uw\in E(\widehat{G}))} \\
\end{eqnarray*}

The graph $G$ is a $(2qtd+1)$-degree bounded bipartite graph on $2|V(\widehat{G})|$ vertices.
Observe that edges in $G$ of the form $(u,1)(u,2)$ correspond to a vertex in $\widehat{G}$.
Thus, a (semi) induced matching in $G$ whose edges are
in this form corresponds to an independent set in $\widehat{G}$.
Although this is not the case for every (semi) induced matching $\mset$ in
$G$, we will show that we can extract a (semi) induced matching $\mset'$
from $\mset$ in such a way that $\mset'$ maps to an independent set in $G$,
and $|\mset'|\geq \Omega(|\mset|)$.

\subsection{Analysis}

Now, we analyze our reduction.
Before proceeding, we prove some useful properties of dispersers.
The next lemma shows the bounds on the size of a $\sigma$-semi-induced
matching in a disperser.


\begin{lemma}[Disperser Lemma]
\label{lem: properties of dispersers}
Every $(d,\gamma)$-disperser $H=(O,Z,E)$ on $2n$ vertices has the
following properties.
\begin{itemize}

\item For any independent set $S$ of $H$, $S$ cannot contain more than 
  $\gamma n$ vertices from both $O$ and $Z$, i.e.,
  \[ 
    \min (|S\cap O|,|S\cap Z|) \leq \gamma n
  \]

\item For any permutation (ordering) $\sigma$ of the vertices of $H$, 
  the graph $\bipp[H]=(U,W,F)$ obtained by transforming $H$ into a bipartite graph (using only edges of type $E_2$)
  contains no $\sigma$-semi induced matching of size more than 
  $4\gamma n$, i.e.,
  \[
    \sinducesigma{\sigma}{\bipp[H]} \leq 4\gamma n
  \]
\end{itemize}
\end{lemma}

\begin{proof}
The first property follows from the definition of the
$(d,\gamma)$-disperser $H$. 
That is, letting $X=S\cap O$ and $Y=S\cap Z$, 
if $|X|,|Y|> \gamma n$, then we must have an edge $xy\in E(H)$ joining
some vertex $x\in X$ to some vertex $y\in Y$, 
contradicting the fact that $S$ is an independent set in $H$.

Next, we prove the second property. 
Consider the set of edges $\mset$ that form a $\sigma$-semi-induced matching in $\bipp[G]$.
We claim that $|\mset| \leq 4\gamma n$.
By way of contradiction, assume that $|\mset| > 4 \gamma n$.
Observe that, for each edge $(u,1)(v,2)\in \mset$, either 
(1)~$u\in O$ and $v\in Z$ or (2)~$v\in O$ and $u\in Z$.
Since the two cases are symmetric,
we analyze only the set of edges of the first case, denoted by $\widehat{\mset}$.
Also, we assume wlog that at least half of the edges of $\mset$ are
in $\widehat{\mset}$;
thus, $|\widehat{\mset}| \geq |\mset|/2 > 2\gamma n$.

Let us denote by $V(\widehat{\mset})$ the set of vertices that are adjacent to some edges in $\widehat{\mset}$. To get a contradiction, we prove the following claim.

\begin{claim}
\label{clm:two-unmatched-pieces}
There are two subsets $X \subseteq U \cap V(\widehat{\mset}): |X| = \gamma n$ and $Y = W \cap V(\widehat{\mset}) : |Y| \geq \gamma n$ such that $\sigma(x) < \sigma(y)$, for any $x \in X$ and $y \in V(\widehat{\mset}) \setminus X$.
Moreover, there is no $\widehat{\mset}$-edge between vertices in $X$ and $Y$.  
\end{claim} 

We first argue that the second property follows from Claim~\ref{clm:two-unmatched-pieces}: If there were such two sets $X,Y$, then we can define the ``projection'' of $X$ and $Y$ onto the graph $H$ by $X' = \set{u \in V(H): (u,1) \in X}$ and $Y' = \set{v \in V(H): (v,2) \in Y}$. It must be the case that $X' \subseteq O$ and $Y' \subseteq Z$ (due to the definition of $\widehat{\mset}$), so from the property of disperser, there is an edge in $E(H)$ joining some $x \in X'$ and $y \in Y'$.
This implies that there must be an edge $(x,1)(y,2) \in E(\bipp[H])$ where $x \in X$ and $y \in Y$. 
Also, there are edges $(x,1)(x',2) \in \widehat{\mset}$ and $(y',1)(y,2) \in \widehat{\mset}$ contradicting the fact that $\mset$ is a $\sigma$-semi-induced matching. 
It only remains to prove the claim.  

\begin{proof}[Proof of Claim~\ref{clm:two-unmatched-pieces}]
Recall that we have the ordering $\sigma$ that is defined on the vertices of $\bipp[H]$, not the vertices of $H$. 
We construct $X$ and $Y$ as follows. 
Order vertices in $U \cap V(\widehat{\mset})$ according to the ordering $\sigma$ and define $X$ to be the first $\gamma n$ vertices according to this ordering. 
So, we have obtained $X \subseteq U \cap V(\widehat{\mset})$ with the property that for any $x \in X$ and $y \in V(\widehat{\mset}) \setminus X$, $\sigma(x) < \sigma(y)$. 

Now, we define $Y \subseteq W \cap V(\widehat{\mset})$ as the set of vertices that are not matched by $\widehat{\mset}$ with any vertices in $X$.
Since $|X|= \gamma n$, the number of vertices in $W \cap V(\widehat{\mset})$ that are matched by $\widehat{\mset}$ is only $\gamma n$, so we can choose arbitrary $\gamma n $ vertices that are not matched as our set $Y$.  
\end{proof} 

\end{proof}

As a corollary of Lemma~\ref{lem: properties of dispersers}, 
we relate the independent number the FGLSS graph $\widetilde{G}$ to
the final graph $G$. 

\begin{corollary}
Let $\widetilde{G}$ and $G$ be the graphs constructed as above.
Then, for any permutation (ordering) $\sigma$ of vertices of $G$, 
\[
\alpha(\widehat{G}) 
  \leq \sinducesigma{\sigma}{G}
  \leq \alpha(\widehat{G}) + 4\gamma |V(\widehat{G})|
\]
\end{corollary}

\begin{proof}
Recall that edges $E(G) = E_1 \cup E_2$. 
To prove the inequality on the left-hand-side, consider the set of edges $E_1$.
Observe that edges of $E_1=\set{(v,1)(v,2):u\in V(\widehat{G})}$
correspond to vertices of $\widehat{G}$ as $\widehat{G}$ and
$\widehat{G}$ share the same vertex set.
Let $S$ be an independent set in $\widehat{G}$. 
We claim that the set $E_S = \set{(u,1)(u,2): u \in S}$ must be an induced matching in $G$, and this would immediately imply the first inequality: 
Assume that there was an edge $(u,1)(v,2) \in E(G)$ for some $(u,1)(u,2), (v,1)(v,2)\in E_S$, so we must have that $u,v \in S$ and that  $uv \in E(\widehat{G}) \subseteq E(\widehat{G})$.
This contradicts the fact that $S$ is an independent set.

Next, we prove the inequality on the right-hand-side.
Let $\mset$ be a $\sigma$-semi-induced matching in $G$.
We decompose $\mset$ into $\mset = \mset_1 \cup \mset_2$.
By the argument similar to the previous paragraph, it is easy to see that $|\mset_1| \leq \alpha(\widehat{G})$: From the set $\mset_1$, we can define a set $S \subseteq V(\widetilde{G})$ by $S= \set{u \in V(\widetilde{G}): (u,1)(u,2) \in \mset_1}$, and $S$ must be an independent set in $\widehat{G}$; otherwise, if there is an edge $uv \in E(\widehat{G})$ for $u, v \in S$, then we would have edges $(u,1)(v,2), (v,1)(u,2)\in E(G)$, contradicting to the fact that $\mset_1$ is $\sigma$-semi-induced matching.  

It is sufficient to show that $|\mset_2| \leq 4 \gamma |V(\widehat{G})|$.
We do so by partitioning $\mset_2$ into $\mset_2 = \bigcup_{j=1}^{N} \mset^j_2$ where $\mset_2^j = \set{(u,1)(v,2) \in \mset_2, uv \in E(H_j)}$ (since $\widehat{G}$ is the union of edges of subgraphs $H_j$). 
Each set $\mset_2^j$ must be a $\sigma_j$-induced matching for the ordering $\sigma_j$ obtained by projecting $\sigma$ onto the vertices of $B[H_j]$, so we can invoke Lemma~\ref{lem: properties of dispersers} to bound the size of $\mset_2^j$, i.e., $|\mset_2^j| \leq 4 \gamma n_j$.  
Summing over all $j$, we have $$|\mset_2| \leq \sum_{j=1}^N |\mset_2^j| \leq \sum_{j=1}^N 4\gamma n_j \leq 4 \gamma q t |V(\widehat{G})|$$  

The last inequality follows because of basic counting arguments.
Each vertex belongs to exactly $qt$ subgraphs $H_j$, so if we sum $n_j$ over all $j=1,2,\ldots,N$, we get $\sum_{j=1}^N n_j = qt |V(\widehat{G})|$. 
\end{proof}


\paragraph{Completeness and Soundness:}
The completeness and soundness proofs are now easy. 
In the \yi, $\alpha(\widehat{G}) \geq c\cdot M$ implies that 
$\sinducesigma{\sigma}{G} \geq c\cdot M$, and 
in the \ni, the fact that $\alpha(\widehat{G}) \leq s\cdot M+ \gamma q t w M$ implies that 
$\sinducesigma{\sigma}{G} 
  \leq  s\cdot M + 5\gamma q t w M$.

%
Now, we choose $\gamma = s/(5q t w)$ and this would give $d=O(\frac{1}{\gamma} \log \frac{1}{\gamma})=O((w q t/s)\log(wq t/s))$. 
Then we have the final graph $G$ with the following properties: 

\begin{center}
\begin{tabular}{|c|c|c|}
\hline
{\bf Number of Vertices} & {\bf Degree} & {\bf Hardness Gap}\\ 
\hline
$n = 2 w M$ & 
$\Delta = (2dq+1)$ &
$g = \displaystyle\frac{c\cdot M}{s\cdot M + 5\gamma \cdot w M}
   \geq \displaystyle\frac{c}{2s}$ \\
\hline
\end{tabular}
\end{center}

Substituting $c,s,w,q,M$ as in
Theorem~\ref{thm: CSP}, we get
\begin{itemize}
\item The degree $\Delta=O(t^2k^42^{t(k^2+2k-1)}) = 2^{t(k^2+ \Theta(k))} = 2^{t(1/\epsilon^2 + \Theta(1/\epsilon))}$ 

\item The number of vertices $|V(G)| = 2^{t(k^2)} N^{1+O(\epsilon)} = \Delta^{1+O(\epsilon)} N^{1+O(\epsilon)}$

\item The hardness gap $g \geq 2^{t(k^2-1)} \geq  \Delta^{1-O(\epsilon)}$   
\end{itemize}

\paragraph{Success probability of the disperser construction:} Notice that the failure probability of the disperser construction given in Lemma~\ref{lemma: disperser} is large when $N_i \gamma$ is small. 
In our case, we have $N_i \geq 2^{tk^2} N^{\delta}$ and $\gamma \geq 2^{-t(k^2 +O(k))}$. 
So, we are guaranteed that $N_i \gamma \geq N^{\delta} 2^{-O(tk)} = 2^{\delta \log N - O(tk)}$.
As long as $t\leq O(\delta \epsilon) \log N$, we would be guaranteed that $N_i \gamma \geq N^{\delta/2}$, so the failure probability in \Cref{lemma: disperser} is at most $2^{-N^{\delta/2}}$.  
This allows us to apply union bound over all variables $x_j$ in the CSP and conclude that the construction is successful with high probability.  
If we appropriately pick $\delta = \Theta(\epsilon)$ and $t \leq 5 \epsilon^2 \log N$, then we obtain Theorem~\ref{thm: hardness of semi-induced matching}.

\subsection{Subexponential Time Approximation Hardness for the Maximum Independent Set and Induced Matching Problems}

Here we present hardness results that show a trade-off between running-time and approximation-ratio. Roughly speaking, we obtain the following results under the Exponential Time Hypothesis: any algorithm that guarantees an approximation ratio of $r$ for the maximum independent set problem and the maximum bipartite induced matching problem on bipartite graphs, for any $r\geq r_0$ for some constant $r_0$, must run in time at least $2^{n^{1-\epsilon}/r^{1+\epsilon}}$.
This almost matches the upper bound of $2^{n/r}$ given by Cygan~et~al.~\cite{CyganKPW08} for the case of the maximum independent set problem and by our simple algorithm given in \Cref{sec:algo-ind} for the case of the bipartite induced matching problem. 
These results are obtained as by-products of \Cref{thm: hardness of semi-induced matching,thm: hardness of independent set}.

Theorem \Cref{thm: hardness of independent set} implies almost immediately the following corollary.

\begin{theorem}\label{thm:sub-exp-hardness-indset}
\bundit{I removed the footnote since we have a discussion in a paragraph below.}
Consider the maximum independent set problem on an input graph $G=(V,E)$.
For any $\epsilon>0$ and sufficiently large $r\leq |V(G)|^{1/2-\epsilon}$, every algorithm that guarantees an approximation ratio of $r$ must run in time at least $2^{|V(G)|^{1-2\epsilon}/r^{1+4\epsilon}}$ unless the ETH is false. 
\end{theorem}
\begin{proof}
The intuition is very simple. \Cref{thm: hardness of independent set} can be seen as a reduction from SAT of size $N$ to the independent set problem whose instance size is, roughly, $|V(G)| = N r$ where $r$ is the approximability gap for the independent set problem (ignoring the small exponent $\epsilon$ in the theorem). (In fact, the graph resulting from the reduction is an $r$-degree-bounded graph as we will set $r\approx d$.)
It is immediate that getting a running time of $2^{o(|V(G)|/r)}$ for the maximum independent set problem is equivalent to getting a running time of $2^{o(N)}$ for SAT (since $N\approx |V(G)|/r$), contradicting the ETH. Below we give a formal proof. 

%

Assume for a contradiction that there is an algorithm $\aset$ that obtains $r$-approximation in time $2^{|V(G)|^{1-2\epsilon}/r^{1+4\epsilon}}$ for some $r \leq |V(G)|^{1/2-\epsilon}$ and a small constant $\epsilon >0$.
Then we can use the algorithm $\aset$ to decide the satisfiability of a given SAT formula $\phi$ as follows. 
First, we invoke the reduction in \Cref{thm: hardness of independent set} on the SAT formula $\phi$ to construct a graph $G=(V,E)$ with parameters $t, \epsilon$ such that $d= 2^{t(1/\epsilon^2+\Theta(1/\epsilon))}=r^{1+3\epsilon}$. 
Notice that the value of $t$ is at most $t \leq   2\epsilon^2 \log r \leq  5 \epsilon^2 \log N$, so the reduction is guaranteed to be successful with high probability. 

Since $d= r^{1+3\epsilon}$, we have $r < d^{1-\epsilon}$, which means we can use the algorithm $\aset$ to distinguish between \yi and \ni in time $2^{|V(G)|^{1-2\epsilon}/r^{1+4\epsilon}} < 2^{N^{1-\epsilon}}$. 
When plugging in the values $|V(G)| \leq N^{1+\epsilon} d^{1+\epsilon}$ and $r^{1+4\epsilon} \geq d$, this violates the ETH.
\end{proof} 

Notice that, since $t$ in~\Cref{thm: hardness of independent set} cannot be chosen beyond $5 \epsilon^2 \log N$, we have no flexibility of making $r$ arbitrary close to $|V(G)|^{1-\epsilon}$.
However, this can be easily fixed by slightly modifying the proof of \Cref{thm: hardness of independent set} and leaving the flexibility of choosing parameter $\delta$ as discussed in \Cref{sec: PCP}. 
Since this is not necessary for the hardness of the $k$-hypergraph pricing problem and will make the proof in  \Cref{sec: PCP} more complicated, the detail is omitted.  

The subexponential time hardness of approximating the maximum induced matching problem can be proved analogously, but we need \Cref{thm: hardness of semi-induced matching} instead of \Cref{thm: hardness of independent set}.

\begin{theorem} \label{thm:sub-exp-hardness-ind-matching}
Consider the maximum induced matching problem on a bipartite graph $G=(U,V,E)$.
For any $\epsilon>0$ and sufficiently large $r\leq |V(G)|^{1/2-\epsilon}$, every algorithm that guarantees an approximation ratio of $r$ must run in time at least $2^{|V(G)|^{1-2\epsilon}/r^{1+4\epsilon}}$ unless the ETH is false. 
\end{theorem}

We skip the proof of Theorem~\ref{thm:sub-exp-hardness-ind-matching} as it follows the same line as the proof of Theorem~\ref{thm:sub-exp-hardness-indset}.

\subsection{Subexponential-Time Approximation Algorithm for Induced Matching}
\label{sec:algo-ind} 
In this section, we show $r$-approximation algorithms which run in $2^{n/r} \poly(n)$ time for the case of bipartite graphs and $2^{(n/r)\log \Delta} \poly(n)$ time for the general case, where $\Delta$ is the maximum degree. 
These running times are nearly tight, except that the second case incurs an extra $O(\log \Delta)$ factor in the exponent. 
We leave the question whether this extra term is necessary as an open problem.

\paragraph{Bipartite Graphs} For the case of bipartite graphs, we prove the following theorem. 

\begin{theorem}[Algorithm on Bipartite Graphs] \label{thm:subexp-algo-induced-matching bipartite}
For any $r \geq 1$, there is an $r$-approximation algorithm for the maximum bipartite induced matching problem that runs in time $2^{n/r} \poly(n)$ where $n$ is the size of the input graph. 
\end{theorem} 

To prove \Cref{thm:subexp-algo-induced-matching bipartite}, we will need the following lemma, which says that we can compute maximum induced matching in bipartite graphs in time $2^{n'}$ where $n'$ is the cardinality of the smaller side of the graph. 

\begin{lemma} 
For any bipartite graph $G=(U,W,E)$, there is an algorithm that returns a maximum induced matching in $G$ and runs in $2^{\min(|U|, |W|)} \poly |V(G)|$ time.
\end{lemma} 

\begin{proof} We assume without loss of generality that $|U| \leq |W|$. 
We first need the characterization of the existence of an induced matching in terms of {\em good neighbors}, defined as follows. 
Given a subset $U'\subseteq U$ and a fixed $u \in U'$, we say that $w \in W$ is a {\em $U'$-good neighbor of $u$} if there is no other $u' \in U'$ (where $u'\neq u$) such that $u' w \in E$. 
We first note the observation that $U' \subseteq U$ forms end-vertices of some induced matching $\mset'$ if and only if every vertex in $U'$ has a $U'$-good neighbor in $W$. To see this, if $U' \subseteq U$ is a set of end-vertices of matching $\mset'$, then it is clear that for each $uw \in \mset$, the vertex $w$ is a $U'$-good neighbor.  
For the converse, for any $u \in U'$, let $w_u \in W$ be a $U'$-good neighbor of $u$. 
Then $\set{u w_u: u \in U'}$ must form an induced matching.

Now, using this observation, we compute a maximum induced matching in $G$ as follows. 
For each possible subset $U' \subseteq U$, we check whether vertices in $U'$ can be end-vertices of any induced matching. 
This can be done in $\poly(|V(G)|)$ time (simply by checking the existence of $U'$-good neighbors). 
We finally return the maximum-cardinality subset $U'$ and its corresponding induced matching $\mset'$.
\end{proof} 

From this lemma, given an input graph $G=(U,W,E)$, we partition the vertices of $U$ into $r$ sets $U_1,\ldots, U_{r}$ in a balanced manner and define $G_i = G[U_i \cup W]$, i.e., $G_i$ is an induced subgraph on vertices $U_i \cup W$.
Our algorithm simply invokes the above lemma on each graph $G_i$ to obtain an induced matching $\mset_i$, and finally we return $\mset_{i^*}$ with maximum cardinality among $\mset_1,\ldots, \mset_r$.
Since $|U_i| \leq \ceil{n/r}$, the running time of our algorithm is at most $2^{\ceil{n/r}} \poly n$.  
The following lemma implies that $\mset_{i^*}$ is an $r$-approximation and feasible in $G$, thus completing the proof.

\begin{lemma}
\label{lmm:fact-induced-matching}
The following holds on $G$ and its subgraphs $G_i$.
\begin{itemize}
\item Any induced matching $\mset_i$ in $G_i$ is also an induced matching in $G$. 
\item $\sum_{i=1}^r \induce{G_i} \geq \induce{G}$ 
\end{itemize}
\end{lemma}    

\begin{proof}
First, we prove the first fact.
If $\mset_i$ is not an induced matching in $G$, then there must be two edges $uv, ab \in \mset_i$ that are joined by some edge $e$ in $G$. 
But, since $u,v,a,b \in U_i \cup W$, the edge $e$ must also be present in $G_i$, contradicting to the fact that $\mset_i$ is an induced matching in $G_i$.  

Next, we prove the second fact.
Let $\mset$ be a maximum induced matching in $G$.  
For $i=1,2,\ldots,r$, define $\mset_i=\mset\cap E(G_i)$.
It is clear that each $\mset_i$ is an induced matching in $G_i$.
Thus, it follows immediately that $\sum_{i=1}^r \induce{G_i} \geq \sum_{i=1}^r |\mset_i| = |\mset|$.  
\end{proof}

\paragraph{General Graphs} 
We note that almost the same running time can be obtained for the case of general graphs, except that we have an extra $\log \Delta$ factor in the exponent, where $\Delta$ is the maximum degree.

\begin{theorem}[Algorithm on General Graphs] \label{thm:subexp-algo-induced-matching}
For any $r \geq 1$, there is an $r$-approximation algorithm for the maximum induced matching problem that runs in time $2^{(n/r)\log \Delta} \poly(n)$ where $n$ is the size of the input graph and $\Delta$ is the maximum degree. 
\end{theorem} 

To prove Theorem~\ref{thm:subexp-algo-induced-matching}, we give the following algorithm.\bundit{I think it is better to present an algorithm first and then give an analysis later.}
Our algorithm takes as input a graph $G=(V,E)$ on $n$ vertices and a parameter $r$.
We first partition $V$ arbitrarily into $V = \bigcup_{i=1}^r V_i$ such that the size of $V_i$'s are roughly equal, 
i.e., $|V_i| = \lfloor n/r \rfloor$ or $|V_i|=\lfloor n/r \rfloor + 1$.
For each $i=1, \ldots, r$, we find a maximum-cardinality subset of edges $M_i$ such that $M_i$ is an induced matching in $G$ and every edge in $M_i$ has at least one end-vertex in $V_i$.
We implement this step by checking every possible subset of edges: We choose one edge incident to each vertex in $V_i$ or choose none and then check whether the set of chosen edges $F$ is an induced matching in $G$, and we pick the set $F$ that passes the test with maximum cardinality as the set $M_i$.
Finally, we choose as output the set $M_i$ that has maximum-cardinality over all $i=1,2,\ldots,r$.

It can be seen that the running time of our algorithm is
$O(\Delta^{n/r}\cdot\poly(n))=O(2^{(n/r)\log{\Delta}}\poly(n))$, where $\Delta$ is the maximum degree of $G$.
For the approximation guarantee, it suffices to show that 
\[\induce{G} \leq \sum_{i=1}^r|M_i| \leq r \cdot \induce{G}.\]
So, we shall complete the proof of Theorem~\ref{thm:subexp-algo-induced-matching} by proving the above inequalities as in the following decomposition lemma.

\begin{lemma}
Consider any graph $G=(V,E)$.
Let $V_1\cup V_2 \cup \ldots \cup V_r$ be any partition of $V$.
For $i=1,2,\ldots,r$, let $M_i$ be a set of edges with maximum-cardinality such that $M_i$ is an induced matching in $G$ and every edge in $M_i$ has at least one end-vertex in $V_i$.
Then $\induce{G} \leq \sum_{i=1}^r|M_i| \leq r \cdot \induce{G}$.
\end{lemma}

\begin{proof}
Let $\mset$ be any maximum induced matching in $G$.
Then, clearly, $|M_i|\leq |\mset| = \induce{G}$ for all $i=1,2,\ldots,r$ because $M_i$ is an induced matching in $G$.
Thus, $\sum_{i=1}^r|M_i| \leq r \cdot\induce{G}$, proving the second inequality.

For $i=1,2,\ldots,r$, define $\mset_i$ to be a subset of $\mset$ such that each edge in $\mset_i$ has at least one end-vertex in $V_i$, and $\induce{G} = \sum_{i=1}^r|\mset_i|$.
%
%
By the maximality of $M_i$, we have $|\mset_i|\leq |M_i|$, for all $i=1,2,\ldots,r$.
Thus, $\induce{G} \leq \sum_{i=1}^r |M_i|$, completing the proof.
\end{proof}

%% file: hardness-pricing.tex
\section{Hardness of $k$-Hypergraph Pricing Problems}
\label{sec:hardness}\label{sec:hardness_pricing}



In this section, we prove the hardness of the $k$-hypergraph pricing problems, as stated formally in the following theorem. 
Throughout this section, we use $n$ and $m$ to denote the number of items and consumers respectively. 
We remark the difference between $n$ (the number of items in the pricing instance) and $N$ (the size of 3SAT formula).

\begin{theorem}\label{thm:main hardness}
Unless \NP = \ZPP, for any $\epsilon >0$, there is a universal constant $k_0$ (depending on $\epsilon$) such that the $k$-hypergraph pricing problem for any constant $k >k_0$ is $k^{1-\epsilon}$ hard to approximate. 
Assuming Hypothesis~\ref{hypo:ETH}, for any $\epsilon >0$, the $k$-hypergraph pricing problem is hard to approximate to within a factor of $\min(k^{1-\epsilon}, n^{1/2-\epsilon})$.
\end{theorem}





\paragraph{Proof Overview and Organization} 
For any $k$-hypergraph pricing instance $(\cset, \iset)$, we denote by $\opt(\cset,\iset)$ the optimal possible revenue that can be collected by any price function.  
The key in proving \Cref{thm:main hardness} is the connection between the hardness of the semi-induced matching and the $k$-hypergraph pricing problem, as stated in the following lemma, which will be proved in \Cref{sec:indmatching-to-hypgraph-pricing}.

\begin{lemma}[From Semi-induced Matching to Pricing; Proof in \Cref{sec:indmatching-to-hypgraph-pricing}] \label{thm:semi-ind-to-pricing}
There is a randomized reduction that, given
a bipartite graph $G=(U,V,E)$ with maximum degree $d$, outputs an instance
$(\cset, \iset)$ of the $k$-hypergraph pricing problem such that, with
high probability,
\[
(6\ln{d}/\ln\ln{d})\sinduce{G} \geq \opt(\cset, \iset) \geq \induce{G}
\]
The number of consumers is $|\cset| = |U| d^{O(d)}$ and the number of items is $|\iset| =
|V|$. Moreover, each consumer $c \in \cset$ satisfies $|S_c| =d$. The running time of this reduction is $\poly(|\cset|, |\iset|)$.
\end{lemma}

We remark that using the upper bound for $\opt(\cset,\iset)$ in terms of $\sinduce{G}$ instead of $\induce{G}$ seems to be necessary. 
That is, getting a similar reduction with a bound $\opt(\cset, \iset) = \tilde O(\induce{G})$ may not be possible.  

Combining the above reduction in \Cref{thm:semi-ind-to-pricing} with the hardness of the induced and semi-induced matching problems in \Cref{thm: hardness of semi-induced matching} (\Cref{sec: hardness of induced matching}) leads to the following intermediate hardness result, which in turn leads to all the hardness results stated in \Cref{thm:main hardness}.

\begin{lemma}[Intermediate Hardness; Proof in \Cref{sec:proof of thm:main:hardness}]
\label{thm:main:hardness}
Let $\epsilon>0$ be any constant. There is a universal constant $d_0 = d_0(\epsilon)$ such that the following holds.
For any function $d(\cdot)$ such that $d_0 \leq d(N) \leq
N^{1-\epsilon}$, there is a randomized algorithm that transforms an
$N$-variable 3SAT formula $\phi$ to a $k$-hypergraph pricing instance $(\cset, \iset)$ such that:

\begin{itemize}
    \item For each consumer $c$, $|S_c| = d(N)$.
   \item The algorithm runs in time $\poly(|\cset|, |\iset|)$.

   \item $|\cset| \leq d^{O(d)} N^{1+\epsilon}$ and $|\iset| \leq N^{1+\epsilon} d^{1+\epsilon}$.

   \item There is a value $Z$ such that (\yi) if $\phi$ is satisfiable, then $\opt(\cset, \iset) \geq Z$; and (\ni) if $\phi$ is not satisfiable, then $\opt(\cset, \iset) \leq Z/d^{1-\epsilon}$.
\end{itemize}
\end{lemma}

In the next section, we prove the reduction in \Cref{thm:semi-ind-to-pricing}. We will prove the above intermediate hardness result (\Cref{thm:main:hardness}) in \Cref{sec:proof of thm:main:hardness} and then the main hardness results of this section (\Cref{thm:main hardness}) in \Cref{sec:thm:main hardness}. 


\subsection{From Semi-Induced Matching to Pricing Problems (Proof of \Cref{thm:semi-ind-to-pricing})}
\label{sec:indmatching-to-hypgraph-pricing}

Here we prove \Cref{thm:semi-ind-to-pricing} by showing a reduction from the semi-induced matching problem on a $d$-degree bounded bipartite graph to the {\em $k$-hypergraph pricing} problem. This reduction is randomized and is guaranteed to be successful with a constant probability. 

\subsubsection{The Reduction}

Let $G=(U,V,E)$ be a bipartite graph with maximum degree $d$. Assume without loss of generality the following, which will be important in our analysis. 
\begin{assumption}\label{assume:UleV}
$|U|\leq |V|$. 
\end{assumption}
Notice that we always have $\sinduce{G} \geq \induce{G} \geq  |U|/d$.
For each vertex $u$ of $G$, we use $N_G(u)$ to denote the set of
neighbors of $u$ in $G$.
If the choice of a graph $G$ is clear from the context, then we will omit the
subscript $G$.
Our reduction consists of two phases.

\paragraph{Phase 1: Coloring} We color each vertex $u\in U$ of $G$ by uniformly and independently choosing a random color from $\{1,2,\ldots,d\}$. We denote by $U_i \subseteq U$, for each $i=1,2,\ldots,d$, the set of left vertices
that are assigned a color $i$. We say that a right vertex $v \in V$ is {\em highly congested} if there is some $i \in [d]$ such that $|N_G(v) \cap U_i| \geq  3\ln{d}/\ln\ln{d}$; i.e., $v$ has at least $ 3\ln{d}/\ln\ln{d}$ neighbors of the same color. 
Let $V_{high} \subseteq V$ be a subset of all vertices that are highly congested and $V' = V \setminus V_{high}$. Thus, $V'$ is the set of vertices in $V$ with highly congested vertices thrown away.\danupon{I removed what was before the first paragraph almost completely. See the commented text below.}
Let $G'$ be a subgraph of $G$ induced by $(U,V',E)$. The following property is what we need from this phase in the analysis in \Cref{sec:hardness analysis to pricing}.\danupon{I changed the statement below a bit. Please check.}

\begin{lemma} \label{lem:sinduce after coloring}
With probability at least $1/2$, 
\[\induce{G'} \geq (1-2/d)\induce{G} ~~~\mbox{and}~~~ \sinduce{G'} \geq (1-2/d)\sinduce{G}.\]
In particular, for $d\geq 4$, $\induce{G'} \geq \induce{G}/2$ and $\sinduce{G'} \geq \sinduce{G}/2$ with probability at least 1/2. 
\end{lemma} 
\begin{proof}
First, consider any vertex $v\in V$. We claim that vertex $v$ is highly congested with probability at most $1/d$. To see this, we map our coloring process of neighbors of $v$ to the {\em balls and bins} problem, where we think of each $u\in N_G(v)$ as a ball $b_u$ and a color $c$ as a bin $B_c$. Coloring a vertex $u\in N_G(v)$ with color $c$ corresponds to putting a ball $v$ to a bin $c$. By the well-known result on this problem (e.g., see \cite[Section 5.2]{MitzenmacherUpfalBook} and \cite{Gonnet81}), we know that with probability at least $1-1/d$, all bins have
at most $3 \ln{d}/\ln\ln{d}$ balls; i.e., $|N_G(v) \cap U_i| \leq  3\ln{d}/\ln\ln{d}$. Thus, $v$ is highly congested with probability at most $1/d$ as claimed.

Consider a maximum induced matching $M=(A, B, F)$ of $G$, where $A\subseteq U$, $B\subseteq V$ and $F\subseteq E$; i.e., $|A|=|B|=|F|= \induce{G}$. Let $M'=(A', B', F')$ be the subgraph of $M$ obtained by removing vertices in $V_{high}$, i.e., $A'=A$ and $B'=B\setminus V_{high}$. Note that edges in $M'$ give an induced matching of $G'$ of size $|B'|$, i.e.,
\[
\induce{G'}\geq |B'|. 
\]
Note that since every vertex $v\in V$ is in $V_{high}$ with probability at most $1/d$, we have $\mathbb{E}[|B\cap V_{high}|]\leq |B|/d$, and thus we can use Markov's inequality to get
\[ \mathbb{P}[|B\cap V_{high}|\geq 2|B|/d]\leq 1/2.
\]
Thus, 
\[ \mathbb{P}[|B'|\leq (1-2/d)|B|]\leq 1/2.
\]
It follows that $\induce{G'}\geq |B'| \geq (1-2/d)|B|= (1-2/d)\induce{G}$ with probability at least $1/2$. This proves the first inequality in the statement.
Proving the second inequality uses exactly the same argument except that we let $M$ be a maximum semi-induced matching and note that edges in $M'$ give a semi-induced matching of $G'$ of size $|B'|$. This completes the proof of \Cref{lem:sinduce after coloring}. 
\end{proof}

\paragraph{Phase 2: Finishing}
Now, we have a coloring of left vertices (in $U$) of $G$ with desired properties.
We will construct an instance of the $k$-hypergraph pricing problem in both \UDP and \SMP models as follows.
For each vertex $v \in V'$, we create an item $I(v)$.
For each vertex $u \in U_i$, we create $d^{3i}$ consumers;
we denote this set of consumers by $\cset(u)$.
We define the budget of each consumer $c\in\cset(u)$ where $u\in U_i$
to be $B_c=d^{-3i}$ and define $S_c = \set{I(v):v \in N_G(u)}$, so it is immediate that $|S_c| \leq d$. 
To recap the parameters of our construction, we have
\begin{itemize}
\item The set of items $\iset = \set{I(v) : v\in V'}$.
\item The set of consumers $\cset = \bigcup_{u\in U}\cset(u)$,
      where $|\cset(u)|=d^{3i}$ for $u \in U_i$.
\item A budget $B_u = \displaystyle\frac{1}{d^{3i}}$
      for each consumer $u\in U_i$.
\item A set of desired items $S_c = \set{I(v): v\in N_G(u)}$
      for each customer $c\in \cset$. 
      (Note that $|S_c| \leq d$.)\bundit{Just added: Aug 12, 2013}
\end{itemize}

This completes the description of our reduction.
\bundit{Just added: Aug 12, 2013}
Note that, in the $k$-hypergraph formulation, we have $\iset$ as a set of vertices, $\cset$ as a set of hyperedges and $k=d$ (since $|S_c|\leq d$ for all $c\in\cset$).


%

\subsubsection{Analysis}\label{sec:hardness analysis to pricing}

\paragraph{Completeness:} 
We will show that the profit we can collect is at least $\induce{G'} \geq \induce{G}/2$ (by \Cref{lem:sinduce after coloring}). 
Let $\mset$ be any induced matching in the graph $G'$.
For each item $I(v)$ with $uv \in \mset$ and $u \in U_i$, we set its
price to $p(I(v)) = 1/d^{3i}$.
For all other items, we set their prices to $\infty$ for \UDP and $0$ for \SMP.
Notice that, for each $u \in U_i$ that belongs to $\mset$, any consumer $c \in \cset(u)$ only sees one item of finite
price (that is, $1/d^{3i}$).
So, for \UDP the consumer $c$ must buy the item $I(v)$ and thus contributes
$1/d^{3i}$ to the total profit. Similarly, for \SMP, the consumer $c$ can afford to buy the whole set $S_c$ of total cost $1/d^{3i}$.
Since $|\cset(u)| = d^{3i}$, the profit contributed from each set of
consumers $\cset(u)$ is $1$.
This implies that the total profit we obtain from this price function is
$|\mset|$.

\paragraph{Soundness:} 
Now, suppose that an {\em optimal} price function $p$ yields a profit of
$r$ (for either \UDP or \SMP).
We will show that $\sinduce{G} \geq r \log \log d/ (12 \log d)$. 
The proof has two steps. In the first step, we identify a collection of ``tight consumers'' which roughly correspond to those consumers who pay sufficiently large fraction of their budgets. Then we construct a large semi-induced matching from these tight consumers. 
We say that a consumer $c\in\cset$ is {\em tight} if she spends at
least $1/4d$ fraction of her budget for her desired item. A vertex $u \in U$ is tight if its set of consumers $\cset(u)$ contains a tight consumer. Let $\cset'$ be the set of tight consumers.
\begin{claim}\label{thm:tight consumer}
The profit made only by tight consumers is at
least $r/2$. 
\end{claim}
\begin{proof}
Observe that {\em profitable non-tight} consumers contribute at most $|U|/4d$ to the profit. 
Since we prove in the last section that $r \geq \induce{G'} \geq \induce{G}/2 \geq |U|/2d$ (from~\Cref{assume:UleV} and~\Cref{lem:sinduce after coloring}), the revenue made from non-tight consumers is at most $r/2$.  
\end{proof}

Now, we construct from the set of tight consumers $\cset'$,
a $\sigma$-semi-induced matching in $G$ for some total order $\sigma$.
We define $\sigma$ in such a way that vertices in $U$ is ordered by their colors (increasingly for the case of \UDP and decreasingly for the case of \SMP).
%

\begin{definition}
Let $\sigma$ be a total order of vertices such that vertices in $U_i$ always
precede vertices in $U_j$ if $i <j$ for \UDP (and $i >j$ for \SMP). (Recall that $U_i$ is the set of vertices in $U$ of color $i$.)
\end{definition}

Let $U'=\{u\in U: \cset(u)\cap\cset'\neq\emptyset\}$; i.e., it is the set of left vertices whose $\cset(u)$ contains a tight consumer. 
Note that $|U'| \geq r/2$ by \Cref{thm:tight consumer}.
For {\sf UDP}, we define a set of edges $\mset$ to be such that
an edge $uv$ is in $\mset$ if $u \in U'$ and
a tight consumer in $\cset(u)$ buys an item $I(v)$. 
For {\sf SMP}, $\mset$ is defined to contain $uv$ such that $u \in U'$ and $I(v)$ is the most expensive item for a consumer in $\cset(u)$.
Note that 
\[|\mset|\geq |U'|\geq r/2\,.\]
This collection $\mset$ may not be a $\sigma$-semi-induced matching
and may not even be a matching.
So, we have to remove some edges from $\mset$ so that the resulting
set is a $\sigma$-semi-induced matching.
To be precise, we will next extract from $\mset$ a set of edges $\mset' \subseteq \mset$ that is
a $\sigma$-semi-induced matching with cardinality
$|\mset'| \geq r\log \log d/6 \log d$, implying that $\sinduce{G} \geq r\log \log d/6 \log d$. 



Our intention is to construct $\mset'$ by adding to it one edge from $\mset$ at a time, as long as $\mset'$ is a $\sigma$-semi-induced matching (if adding an edge makes $\mset'$ not a $\sigma$-semi-induced matching, then we will not add it). The order of edges we pick from $\mset$ depends {\em reversely} on $\sigma$, and we will also do this process separately for different colors of left vertices as follows. 
We partition $\mset$ into
$\mset_1\cup \mset_2\cup\ldots\cup \mset_d$,
where 
$$\mset_i=\{uv\in\mset:u\in U_i\};$$ 
i.e.,
$\mset_i$ contains edges $uv$ whose end-vertex $u$ is colored $i$.
Then we construct from each set $\mset_i$ a set of edges $\mset'_i$ as follows.
We process each edge $uv\in \mset_i$ in the {\em reverse} order of $\sigma$; i.e.,  an edge $uv$ is processed before an edge $u'v'$ if $\sigma(u)>\sigma(u')$.
%
%
%
For each edge $uv \in \mset_i$, we remove from
$\mset_i$ all edges $u'v'$ such that $u'$ is adjacent
to $v$.
Then we add $uv$ to the set $\mset'_i$ and proceed to the next edge
remaining in $\mset_i$.
Notice that, each time we add an edge $uv$ to $\mset'_i$, we remove
at most $3 \log d/ \log \log d$ edges from $\mset_i$ because its
end-vertex is not highly congested by the construction of
$\mset_i$.
So, $|\mset'_i| \geq |\mset_i|\log \log d/3 \log d$.
Moreover, it can be seen by the construction that $\mset'_i$ is
a $\sigma$-semi-induced matching.
Finally, define $\mset' = \bigcup_{i=1}^d \mset'_i$.
Then we have that
\[
|\mset'|
         \geq \frac{|\mset|\log \log d}{3 \log d}\,.
\]
We now claim that $\mset'$ is a $\sigma$-semi-induced matching.
Suppose not.
Then there is a pair of edges $uv, u'v' \in \mset'$ such that
$\sigma(u) < \sigma(u')$ and $uv' \in E(G)$. We need two cases to distinguish between the two models of \SMP and \UDP. 

\begin{itemize}
  \item For \UDP, by the construction, the two vertices $u$ and $u'$ must belong to
different color class $U_i$ and $U_j$, respectively, where $i< j$.
Since $uv' \in E(G)$, consumers in $\cset(u)$ are interested in item
$I(v')$, whose prices are $1/d^{3j}$ (which is strictly less than $1/2d^{i+1}$) because $u'$ is a tight index.
But, then $u$ would have never been tight, a contradiction. 

\item For \SMP, the two vertices $u$ and $u'$ belong to $U_i$ and $U_j$ respectively where $i >j$. Since $uv' \in E(G)$, consumers in $\cset(u)$ are interested in item
$I(v')$, whose prices are $1/2d^{3j+1} > 1/d^{3i}$. Then consumers in $\cset(u)$ would not have enough budget to buy their item sets, contradicting the fact that they are tight. 
\end{itemize}

Thus, we have 
\[ \sinduce{G'}\geq  |\mset'|
         \geq \frac{|\mset|\log \log d}{3 \log d}\geq \frac{r\log \log d}{6 \log d}\]
as desired.


\subsection{Intermediate Hardness (Proof of \Cref{thm:main:hardness})}\label{sec:proof of thm:main:hardness}

We prove \Cref{thm:main:hardness} using \Cref{thm: hardness of semi-induced matching} and \Cref{thm:semi-ind-to-pricing}. We restate \Cref{thm: hardness of semi-induced matching} here: 
%

\inducedmatching*

%



To see how to prove \Cref{thm:main:hardness}  by
combining \Cref{thm: hardness of semi-induced matching} with \Cref{thm:semi-ind-to-pricing}, we start from an $N$-bit 3SAT formula $\phi$ and invoke
\Cref{thm: hardness of semi-induced matching}
to obtain a $d$-degree bounded bipartite graph $G=(U,V,E)$.
Then we apply a reduction as in Theorem~\Cref{thm:semi-ind-to-pricing}
and obtain an instance $(\cset, \iset)$ of \UDP or \SMP with
$|\cset| = |V(G)| d^{O(d)} = N^{1+O(\epsilon)}d^{O(d)}$ and
$|\iset| = N^{1+O(\epsilon)}d^{1+O(\epsilon)}$.
It is immediate that the gap between \yi and \ni is $d^{1-2\epsilon}$
for all values of $d$.

Notice that our reduction gives (nearly) tight hardness results for all values of $d$.
The complexity assumptions that we make are different for
different values of $d$.
(Note that $d=2^{t(1/\epsilon^2 +O(1/\epsilon))}$, so our parameter is indeed $t$ and $\epsilon$.)
For example, if $d$ is constant, then our complexity assumption is
$\NP\neq\ZPP$, and if $d=\polylog{N}$, then our
assumption is $\NP\nsubseteq\ZPTIME(2^{\polylog{N}})$.

To see this, consider the size
(which also implies the running time) of our reduction.
Our instance for \UDP (resp., \SMP) has
the number of consumers $m=|\cset| = N^{1+O(\epsilon)}d^{O(d)} \leq N^{1+O(\epsilon)} 2^{d^{1+\epsilon}}$ and
the number of items $n=|\iset| = N^{1+O(\epsilon)}d^{1+O(\epsilon)}$.
Suppose we have an algorithm with a running time of $\poly(n,m)$.
Then we also have a randomized algorithm with a running time of
$\poly(N^{1+O(\epsilon)},2^{d^{1+O(\epsilon)}})$ that solves SAT exactly.
%



\subsection{Main Hardness Results (Proof of \Cref{thm:main hardness})}\label{sec:thm:main hardness}

From the above discussion, we can see that the complexity assumption that we have
to make is $\NP\nsubseteq\ZPTIME(\poly(N,2^{O(d)}))$.
Thus, if $t$ is constant, then $d$ is also a constant, i.e., $d= 2^{O(1/\epsilon^2)}$, and the corresponding complexity assumption is, indeed,
$\NP\neq\ZPP$. 
In this case, we get the hardness of the $k$-hypergraph pricing problems when $k$ is constant (note that $k=d$). 
To be more precise, we have proved that: For any $\epsilon >0$, there is a constant $k_0$ that depends on $\epsilon$ such that the $k$-hypergraph pricing problem is $k^{1-\epsilon}$ hard for any $k \geq k_0$.  

We note that our reduction also implies the hardness of $\Omega(\log^{1-\epsilon} m)$, as proved by Chalermsook et al.~\cite{ChalermsookCKK12}. 
In this case, if the value of $k$ is chosen to be $\polylog{N}$, then the complexity
assumption is $\NP\nsubseteq\ZPTIME(2^{\polylog{N}})$. 
In particular, we can plug in $k= \log^{1/\epsilon} N$, so we have $m = N^{\log^{1/\epsilon} N}$ and the hardness factor of $k^{1-\epsilon} = \log^{1/\epsilon -1} N = \log^{1-O(\epsilon)} m$.

Now, let us incorporate the ETH (\Cref{hypo:ETH}) with our hardness result.
So, we assume that there is no exponential-time (randomized) algorithm
that solves 3SAT.
We choose $t = (\epsilon^2- O(\epsilon)) \log N$, so we have $k = 2^{t(1/\epsilon^2 - O(1/\epsilon)} = 2^{(1-O(\epsilon)) \log N} = N^{1-O(\epsilon)}$. Moreover, $|\iset| = N^{2+ O(\epsilon)}$, and the size of the resulting pricing instance (as well as the running time) is dominated by $k^{O(k)} \leq 2^{N^{1-\epsilon}}$, which is fine (still subexponential time) because we assume ETH.  
Writing $k$ in terms of the number of items, we have that $k=N^{1 - O(\epsilon)}=n^{1/2 - O(\epsilon)}$.
Thus, our $k^{1-\epsilon}$-hardness result rules out
a polynomial-time algorithm that gives
$n^{1/2-\epsilon}$-approximation for \UDP (resp., \SMP),
assuming ETH.

\subsection{Subexponential-Time Approximation Hardness for the $k$-Hypergraph Pricing Problem}

In this section, we present the approximability/running time trade-off for the pricing problems. We note that this hardness is in different catalogs as for the maximum independent set and maximum induced matching problems (shown in \Cref{sec: hardness of induced matching}): our inapproximability result shows that any $n^{\delta}$-approximation algorithm for the $k$-hypergraph pricing problem (both \UDP and \SMP) must run in time at least $2^{(\log m)^{\frac{1-\delta-\epsilon}{\delta}}}$ for any constant $\delta, \epsilon >0$.
This almost matches the running time of  $O\left(2^{(\log{m})^{\frac{1-\delta}{\delta}}\log\log{m}}\poly(n,m)\right)$ presented in \Cref{sec:algo}.

%

\begin{theorem}
\label{thm: pricing-subexpo}
Consider the $k$-hypergraph pricing problem (with either \SMP or \UDP buying rule).
For any $\delta>0$ and a sufficiently small $\epsilon: \epsilon < \epsilon_0(\delta)$,
every $n^{\delta}$ approximation algorithm for \UDP (resp., \SMP) must run in time at least $2^{(\log m)^{\frac{1-\delta-\epsilon}{\delta}}}$ unless the ETH is false.  
\end{theorem}

To see the implication of this theorem, we may try plugging in $\delta =1/3$, and this corollary says that if we want to get $n^{1/3}$ approximation for the pricing problem, it would require the running time at least $m^{\log^{1-\epsilon} m}$, assuming the ETH. 

\begin{proof}

Fix $\delta <1/2$. Let $\epsilon$ be as in \Cref{thm:main:hardness}. 
Assume (for contradiction) that we have an $n^{\delta}$ approximation algorithm $\aset$ for \UDP (resp., \SMP) that runs in time $2^{(\log m)^{\frac{1-\delta-100\epsilon}{\delta}}}$.
We apply a reduction in Lemma~\ref{thm:main:hardness} with $d= n^{\delta(1+10\epsilon)}$.
Then the algorithm $\aset$ can distinguish between the \yi and \ni, thus deciding the satisfiability of SAT. 
Now, we only need to analyze the running time of the algorithm and show that the algorithm runs in time $O(2^{N^{1-\epsilon}})$, which will contradict the ETH. 

Since we have $n \leq d^{1+\epsilon} N^{1+\epsilon}$, by plugging in the value of $d = n^{\delta(1+10\epsilon)}$, we get $n \leq n^{\delta(1 +20 \epsilon)} N^{1+\epsilon}$, implying that $n \leq N^{1+\delta+40 \epsilon}$.  
Now, we also plug in the value of $d$ into $m \leq 2^{d^{1+\epsilon}} N^{1+\epsilon}$, and we get 
\[
  m \leq 2^{n^{\delta(1+20\epsilon )}} \leq 2^{N^{\delta(1+\delta+40 \epsilon)}}
\]
Hence, we have $\log m \leq N^{\delta(1+\delta+40 \epsilon)}$, implying 
the running time of $\log^{\frac{1-\delta-100\epsilon}{\delta}} m \leq N^{1-10\epsilon}$, which is subexponential in the size of SAT instance, contradicting ETH.   
\end{proof}

%% file: algo.tex
\section{Approximation Scheme for $k$-Hypergraph Pricing}
\label{sec:algo}

In this section, we present an approximation scheme for the $k$-hypergraph pricing problem, which works for both \UDP and \SMP buying rules. 
Throughout, we denote by $n$ and $m$ the number of items and the number of consumers, respectively.
For any parameter $\delta$, our algorithm gives  an approximation ratio of $O(n^{\delta})$ and runs in time
$O(2^{(\log{m})^{\frac{1-\delta}{\delta}}\log\log{m}}\poly(n,m))$.

In the underlying mechanism, 
we employ as subroutines an $O(\log{m})$-approximation algorithm for \UDP (resp., \SMP) and an $O((\log{m})^{n})$-time constant-approximation algorithm for \UDP (resp., \SMP) as stated in the following two lemmas. 

\begin{lemma}[\cite{GuruswamiHKKKM05}]
\label{lmm:log-m-algo}
There is an $O(\log{m})$-approximation algorithm for \UDP
(resp., \SMP), where $m=|\cset|$ is the number of consumers.
\end{lemma}

\begin{lemma}
\label{lmm:expotime-algo}
There is a constant-factor approximation algorithm for \UDP
(resp., \SMP) that runs in time $O((\log{nm})^{n}\poly(n,m))$.
\end{lemma}

For the sake of presentation flow, we defer the proof of Lemma~\ref{lmm:expotime-algo} to Section~\ref{sec:algo:constant-factor-large-running-time}.
Now, we present our approximation scheme and its analysis.

\subsection{Approximation Scheme}
\label{sec:algo:approx-scheme}
We exploit a trade-off between the approximation ratio and the running time.
In particular, the $O(\log{m})$-approximation algorithm from \Cref{lmm:log-m-algo}, denoted by $\aset_1$, always runs in polynomial time but yields a bad approximation ratio in terms of $n$ when $n^\delta\ll \log{m}$.
In contrast, the $O((\log{nm})^{n}\poly(n,m))$-time $O(1)$-approximation algorithm from \Cref{lmm:expotime-algo}, denoted by $\aset_2$, has a slow running time but always gives a good approximation ratio.
So, we take advantage of the trade-off between the running time and approximation ratio by selecting one of these two algorithms according to the values of $n^\delta$ and $\log{m}$.
To be precise, our approximation scheme takes as input a set of consumers $\cset$, a set of items $\iset$ and a parameter $\delta: 0 < \delta < 1$.
If the number of items is large, i.e., $n^\delta > \log{m}$, then we apply the $O(\log{m})$-approximation algorithm $\aset_1$.
Otherwise, we partition the set of $m$ items into $n^\delta$ (almost) equal subsets, namely,
$\iset_1,\ldots,\iset_{n^{\delta}}$, and we apply the algorithm $\aset_2$ to each subinstance $(\cset,\iset_i)$, for $i=1,\ldots,n^\delta$.
We then sell to consumers the set $\iset_{i^*}$ that yields a maximum revenue over all $i=1,\ldots,n^\delta$.
%
%
Here the key idea is that one of the sets $\iset_i$ gives a revenue of at least $\opt/n^\delta$ in the optimal pricing, where $\opt$ is the optimal revenue.
So, by choosing the set that maximizes a revenue, we would get a revenue of at least $O(\opt/n^\delta)$ (because $\aset_2$ is an $O(1)$-approximation algorithm).
Since we have two different buying rules, \UDP and \SMP, there is some detail that we need to adjust.
When we assign the prices to all the items, we need to ensure that the consumers will (and can afford to) buy the set of items we choose.
So, we price items in the set $\iset_{i^*}$ by a price function returned from the algorithm $\aset_2$, and we apply two different rules for filling the prices of items in $\iset\setminus\iset_{i^*}$ for the cases of \UDP and \SMP.
In \UDP, we price items in $\iset\setminus\iset_{i^*}$ by $\infty$ to guarantee that no consumers will buy items outside $\iset_{i^*}$.
In $\SMP$, we price items in $\iset\setminus\iset_{i^*}$ by $0$ to guarantee that consumers can afford to buy the whole set of items that they desire (although we get no profit from items outside $\iset_{i^*}$).
The running time of the algorithm $\aset_2$ in general is large, but since $n^\delta < \log{m}$ and each subinstance contains at most $n^{1-\delta}$ items, we are able to guarantee the desired running time.
Our approximation scheme is summarized in Algorithm~\ref{algo:main-algo}.



\begin{algorithm}
\label{algo:main-algo}
\caption{{\sf Pricing}($\cset$,$\iset$,$\delta$)}
\begin{algorithmic}[1]
\If{$n^{\delta}> \log{m}$}
  \State Apply an $O(\log{m})$ approximation algorithm for \UDP
  (resp., \SMP) from \Cref{lmm:log-m-algo}.
  \State \Return The price function $p$ obtained by the
  $O(\log{m})$-approximation  algorithm.
\Else
  \State Partition $\iset$ into $n^\delta$ equal sets, namely
    $\iset_1,\iset_2,\ldots,\iset_{n^\delta}$.
    So, each $\iset_j$ has size $|\iset_j|\leq n^{1-\delta}$.
  \For{$j=1$ to $n^\delta$}
    \State Apply~\Cref{lmm:expotime-algo}
      on the instance $\Pi_j = (\cset, \iset_j)$, i.e., restricting  the set of items to $\iset_j$.
  \EndFor
  \State Choose an instance $\Pi_{j^*}$ that maximizes the revenue
    over all $j=1,2,\ldots,n^{\delta}$.
  \State Let $p$ be the price function obtained by solving an
    instance $\Pi_{j^*}$.
  \State For \UDP (resp. \SMP), set the prices of all the items in
    $\iset \setminus \iset_{j^*}$ to $\infty$ (resp. $0$).
  \State \Return the price function $p$.
\EndIf

\end{algorithmic}
\end{algorithm}

\subsection{Cost Analysis}
\label{sec:algo:cost-analysis}
First, we analyze the approximation guarantee of our algorithm.
If $n^{\delta} > \log{m}$, then our algorithm immediately gives
$O(n^{\delta})$-approximation.
So, we assume that $n^{\delta}\geq O(\log{m})$.
We will use the following two lemmas. 

\begin{lemma}\label{lemma: extension of prices} 
For any instance $(\cset, \iset)$ of \UDP (resp. \SMP), let $\iset'$ be any subset of $\iset$. Let $p'$ be a price function that collects a revenue of $r$ from $(\cset, \iset')$. Then, the price function $p: \iset \rightarrow \R$ obtained by setting $p(i) = \infty$ (resp. $p(i) = 0$) for $i \in \iset \setminus \iset'$ and $p(i) = p'(i)$ for $i \in \iset'$ gives revenue at least $r$ for the instance $(\cset, \iset)$.
\end{lemma}

\begin{proof}
We first prove the lemma for \UDP.
Consider the price function $p'$ that collects a revenue of $r$.
Notice that, under the price $p$, each customer $c \in \cset$ who has positive payment in $p'$ also pays for the same amount in $p$ 
(since items in $S_c \cap (\iset \setminus \iset')$ have infinite prices).

For \SMP, for each customer $c \in \cset$ who pays positive price in $p'$, we have by construction that $\sum_{i \in S_c} p(i) = \sum_{i\in S_c \cap \iset'} p'(i)$.
So, the customer $c$ can still afford the set and pays the same amount as in the subinstance $(\cset, \iset')$.
\end{proof}

The above lemma allows us to focus on analyzing the revenue obtained from the subinstance
$(\cset, \iset_{j^*})$.
Since we apply a constant-factor approximation algorithm to the instance $(\cset, \iset_{j^*})$, it suffices to show that $\opt(\cset, \iset_{j^*}) \geq \opt(\cset, \iset)/n^{\delta}$,
which follows from the following lemma.

\begin{lemma}\label{lemma: decomposition of instance}
Let $q$ be any positive integer. For any set of consumers $\cset$ and any partition of $\iset$ into $\iset = \bigcup_{j=1}^q \iset_j$, the following holds for \UDP (resp., \SMP)
\[
\opt(\cset, \iset) \leq \sum_{j=1}^q \opt(\cset, \iset_j)
\]
\end{lemma}

\begin{proof}
Consider an optimal price function $p^*$ for $(\cset, \iset)$.
Fix some optimal assignment of items to customers with respect to $p^*$.
Now for each $j = 1,\ldots, q$, let $r_j$ be the revenue obtained by function $p^*$ from items in $\iset_j$, so we can write $\sum_{j=1}^q r_j = \opt(\cset, \iset)$. 
Notice that, in each sub-instance $(\cset, \iset_j)$, we can restrict the price function $p^*$ onto the set $\iset_j$ and obtain the same revenue. 
This means that $\opt(\cset, \iset_j) \geq r_j$, implying that 
\[\sum_{j=1}^q \opt(\cset, \iset_j) \geq \sum_{j=1}^q r_j =\opt(\cset, \iset)\]  
as desired. 
\end{proof}

\subsection{Running Time Analysis}
\label{sec:algo:running-time}

If $n^{\delta}>\log{m}$, then our algorithm runs in polynomial-time 
(since we apply a polynomial-time $O(\log{m})$-approximation algorithm).
So, we assume that $n^{\delta}\leq \log{m}$.
In this case, we run an algorithm from \Cref{lmm:expotime-algo} on $n^\delta$ sub-instances having $n^{1-\delta}$ items each. It follows that the running time of this algorithm is
\begin{align*}
 O\left(n^{\delta} (\log{nm})^{n^{1-\delta}}\poly(n^{1-\delta},m)\right)
 &=
 O\left(2^{(\log{m})^{\frac{1-\delta}{\delta}}\log\log{nm}}\poly(n,m)\right). 
\end{align*}

The equality follows since $\log{m}\geq n^{\delta}$ implies that $n^{1-\delta}\leq (\log m)^{(1-\delta)/\delta} $. Thus, for any constant $\delta>0$, our algorithm runs in
quasi-polynomial time.

\subsection{Polynomial-Time $O(\sqrt{n\log{n}})$-Approximation Algorithm}
\label{sec:algo:poly-time-approx}
Now, we will set $\delta$ so that our approximation scheme runs in polynomial-time.
To be precise, we set $\delta$ so that $n^\delta=\sqrt{n\log{n}}$.
It follows that our algorithm yields an
approximation guarantee of $O(\sqrt{n\log{n}})$.
%
The running time of our algorithm is (note that $n^{1-\delta} = \sqrt{\frac{n}{\log{n}}}$)
\begin{align*}
O\left(n^{\delta}\cdot (\log{nm})^{n^{1-\delta}}\poly(n,m)\right)
&=
 O\left(2^{\frac{\sqrt{n}}{\sqrt{\log{n}}}\log\log{nm}}\poly(n,m)\right)\\
&=
 O\left(2^{\frac{\sqrt{n\log{n}}}{\log{n}}\log\log{nm}}\poly(n,m)\right)
\end{align*}

If $\sqrt{n\log{n}} \leq \log{nm}/\log\log{nm}$, then we are done because
the running time of the algorithm will be $O(2^{\log{nm}}\poly(n,m))=\poly(n,m)$.
Thus, we assume that $\sqrt{n\log{n}} > \log{nm}/\log\log{nm}$.
So, we have
\[
\log{n} > \log\left(\frac{\log{nm}}{\log\log{nm}}\right)
        = \log\log{nm} - \log\log\log{nm}
        \geq  \frac{1}{2}\log\log{nm}
\]
This means that $\log{n}/\log\log{nm}\geq 1/2$.
Thus, the running time of our algorithm is
\[
 O\left(2^{\frac{\sqrt{n\log{n}}}{\log{n}}\log\log{nm}}\poly(n,m)\right)
 \leq O\left(2^{2\sqrt{n\log{n}}}\poly(n,m)\right)
 \leq O(2^{\log{m}}\poly(n,m))
 = \poly(n,m)
\]

The last inequality follows since
$n^{\delta} = \sqrt{n\log{n}} \leq \log{m}$.
Thus, in polynomial-time, our approximation scheme yields an approximation ratio of $O(\sqrt{n\log{n}})$ for both \UDP and \SMP.



\subsection{A Constant-Factor Approximation Algorithm with a Running Time of\\ $O((\log{nm})^{n}\poly(n,m))$ (Proof of Lemma~\ref{lmm:expotime-algo})}
\label{sec:algo:constant-factor-large-running-time}

In this section, we present an $O(1)$-approximation algorithm for \UDP and \SMP that runs in $O((\log{nm})^{n}\poly(n,m))$ time.
%
%
Our algorithm reads as input an instance $(\cset,\iset)$ of \SMP (resp., \UDP) and a parameter $\alpha>1$. 
Let $W$ be the largest budget of the consumers in $\iset$, and 
define a set
\[
\pset = \set{W,\frac{W}{\alpha^1},\frac{W}{\alpha^2},\ldots,\frac{W}{\alpha^{\lceil \log_\alpha(\alpha nm) \rceil}},0}
\]
Our algorithm tries all the possible price functions that take values from the set $\pset$ and returns as output a price function $p$ that maximizes the revenue
(over all the sets of price functions $p:\iset\rightarrow\pset$).
It is easy to see that the running of the algorithm is $O(\lceil\log_{\alpha}{nm}+3\rceil^n\poly(n,m))$.
Thus, we can set $\alpha=2+\epsilon$ for some $\epsilon>0$ so that 
the running time is $O((\log{nm})^n\poly(n,m))$.\bundit{ 
Al thought it does not matter much, we have to be precise when saying something about the base of the exponential since $(\log m + 3)^n = \log^n{m} \cdot 3^n$. 
}
We claim that our algorithm gives an approximation ratio of $\alpha^2/(\alpha-1)$ (proved in \Cref{sec:algo:constant-factor-large-running-time:cost-analysis}), thus yielding a constant-factor approximation.

\subsubsection{Cost Analysis} 
\label{sec:algo:constant-factor-large-running-time:cost-analysis}

We will focus on the case of \SMP.
The case of \UDP can be analyzed analogously.
Fix any optimal price function $p^*$, which yields a revenue of $\opt$.
We will construct from $p^*$ a price function $p'$ that takes values from the set $\pset$ by rounding down each $p^*(i)$ to its closest value in $\pset$. 
Since $p'(i)\in\pset$ for all $i\in\iset$, we can use $p'$ to lower bound the revenue that we could obtain from our algorithm.
%
%
For the ease of analysis, we will do this in two steps.
First, we define a price function $p_1$ by setting 
\[
p_1(i) = \left\{\begin{array}{ll}
           0         & \mbox{ if $p^*(i) < \frac{W}{\alpha nm}$ } \\
           p^*(i)     & \mbox{ otherwise}
           \end{array}\right.
\]
In this step, we lose a revenue of at most $\opt/\alpha$.
This is because we have $m$ consumers, and each consumer wants at most $n$ items.
So, the revenue loss is at most $nm\cdot W/(\alpha nm) \leq \opt/\alpha$ (since $\opt \geq W$).
That is, $p_1$ yields a revenue of at least $(1-1/\alpha)\opt$. 
Next, we define a price function $p_2$ from $p_1$ by setting
\[
p_2(i) = \left\{\begin{array}{ll}
           0         & \mbox{ if $p_1(i) = 0$ (i.e., $p^*(i) < \frac{W}{\alpha nm}$) } \\
           \frac{W}{\alpha^j} 
           \mbox{, for some $j:\frac{W}{\alpha^{j+1}} < p^*(i) \leq \frac{W}{\alpha^j}$}
                     & \mbox{otherwise}
           \end{array}\right.
\]
So, $p_2(i)=p_1(i)$ for all $i$ such that $p^*(i) < W/(\alpha nm)$.
Observe that $p_2(i) \geq p_1(i)/\alpha$ for all $i\in\iset$. 
Hence, the revenue obtained from $p_2$ is within a factor of $1/\alpha$ of the revenue obtained from $p_1$.
Thus, $p_2$ yields a revenue of at least 
$(1/\alpha - 1/\alpha^2)\opt$.
Thus, the approximation ratio of our algorithm is $O(\alpha^2/(\alpha-1))$, for all $\alpha>1$. 
In particular, by setting $\alpha=2+\epsilon$ for $\epsilon>0$, we have a $(4+\epsilon^2)$-approximation algorithm.

\subsubsection{Remark: An Exact-Algorithm for \UDP} 
\label{sec:algo:constant-factor-large-running-time:exact}

In this section, we observe that for the case of \UDP, an optimal solution can be obtained in time $O(n!\poly(n,m))$ by using a {\em price ladder constraint}. 
To be precise, the price ladder constraint says that, given a permutation $\sigma$ of items, the price of an item $\sigma(i)$ must be at most the price of an item $\sigma(i+1)$, i.e., $p({\sigma(i)})\leq p({\sigma(i+1)})$, for all $i=1,\ldots,n-1$.
Briest and Krysta~\cite{BriestK11} showed that \UDP with the price ladder constraint (i.e., the permutation is additionally given as an input) is polynomial-time solvable.
Thus, we can solve any instance of \UDP optimally in time $O((n!)\poly(n,m))$ by trying all possible permutations of the items.

%% file: fglss.tex
\section{FGLSS and Dispersers Replacement}\label{sec:FGLSS detail}

Beside the size of the graph as discussed, the FGLSS reduction has another property, which was observed by Trevisan~\cite{Trevisan01}: the FGLSS graph $G$ is formed by a union of $N'$ bipartite cliques where $N'$ is the number of variables of $\phi_2$, i.e., $G=\bigcup_{i=1}^{N'}G_i$.
The dispersers replacement indeed replaces each bipartite clique $G_i=(A_i,B_i,E)$ with a $d$-regular bipartite graph (with a certain property).
This techniques involves the following parameters of $\phi_2$.
\begin{itemize}
\item[(1)] Linearity: Each constraint of $\phi_2$ is linear.
\item[(2)] Min-Degree $\delta$: The minimum number of clauses that each variable participates in.
\item[(3)] Clause-Size $q$: The number of literal in each clause.
\end{itemize}

Let $H=\bigcup_{i=1}^{N'}H_i$, where $H_i=(A_i,B_i,F_i)$ is a disperser, denote the graph obtained after the disperser replacement.
Then the followings are parameters transformation:

{\footnotesize
\begin{tabular}{lclcl}
{\bf Properties of $\phi_2$} & &
{\bf Properties of $G=\bigcup_{i=1}^{N'} G_i$} & &
{\bf Properties of $H=\bigcup_{i=1}^{N'} H_i$}\\
All constraints are linear.
& $\Rightarrow$ &
$\forall i$, $G_i=(A_i,B_i,E_i)$ has $|A_i|=|B_i|$.
& $\Rightarrow$ &
$\forall i$, $H_i=(A_i,B_i,F_i)$ has $|A_i|=|B_i|$.\\
Min CSP Degree $\delta$.
& $\Rightarrow$ &
$\min_{i=1}^{N'} |V(G_i)| \geq \delta$.
& $\Rightarrow$ &
$\min_{i=1}^{N'} |V(H_i)| \geq \delta$.\\
Clause-Size $q$.
& $\Rightarrow$ &
$\Delta(G) \leq q \cdot \max_{i=1}^{N'} \Delta(G_i)$.
& $\Rightarrow$ &
$\Delta(H) \leq q \cdot d$.\\
\end{tabular}
}

We can set $d$ so that $\Delta(H)\approx k$ and get $\Delta$-hardness. The linearity of $\phi_2$ is required because we need each bipartite clique $G_i$ to be balanced. Also, because the construction of dispersers is randomized, we need $\delta\geq N^{\epsilon}$, for some constant $\epsilon:0 < \epsilon < 1$, to guarantee that each disperser can be constructed with high probability for all hardness parameters $k$, which thus means that we can apply the union bound to show the success probability of the whole construction.
(When $k=O(1)$ or $k=\poly(N)$, the property is not needed.)
To guarantee that $\delta\geq N^{\epsilon}$ for all $k$, we modify the CSP instance $\phi_1$ by making $N^{\epsilon}$ copies of each clauses in Step (4).



